\documentclass{article}

\usepackage{arxiv}

\usepackage[utf8]{inputenc} 
\usepackage[T1]{fontenc}    
\usepackage{hyperref}       
\usepackage{url}            
\usepackage{booktabs}       
\usepackage{amsfonts}       
\usepackage{nicefrac}       
\usepackage{microtype}      
\usepackage{graphicx}
\usepackage[square,numbers]{natbib}
\usepackage{doi}

\usepackage{caption}
\usepackage{lineno, dsfont}
\usepackage{amsmath}
\usepackage{dirtytalk}
\usepackage{multirow}
\usepackage{amsthm}
\usepackage{stmaryrd}
\usepackage{tikz}

\usepackage{algorithm}
\usepackage{algorithmic}

\usepackage{graphicx}
\usepackage{graphbox}

\newtheorem{dfn}{Definition}
\newtheorem{thm}{Theorem}
\newtheorem{prop}{Proposition}
\newtheorem{lemma}{Lemma}
\newtheorem{corollary}{Corollary}

\newcommand{\R}[0]{\mathds{R}} 
\renewcommand{\vec}[1]{{\boldsymbol{{#1}}}} 

\DeclareMathOperator*{\argmax}{\arg\!\max}

\title{Pure Bayesian Nash Equilibrium for Bayesian Games with Multidimensional Vector Types and Linear Payoffs}

\author{ 
    \href{https://orcid.org/0009-0000-2571-2889}{
        \includegraphics[scale=0.06]{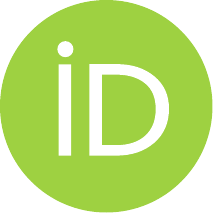}
        \hspace{1mm}Sébastien Huot
    }\\
    Department of Computing\\
    Imperial College London\\
    London, United Kingdom\\
    \texttt{sebastien.huot@protonmail.com} \\
    \And
    \href{https://orcid.org/0000-0002-6211-1991}{
        \includegraphics[scale=0.06]{orcid.pdf}
        \hspace{1mm}Abbas Edalat
    } \\
    Department of Computing\\
    Imperial College London\\
    London, United Kingdom\\
    \texttt{a.edalat@imperial.ac.uk} \\
}

\renewcommand{\headeright}{}
\renewcommand{\undertitle}{}
\renewcommand{\shorttitle}{}

\hypersetup{
pdftitle={Pure Bayesian Nash Equilibrium for Bayesian Games with Multidimensional Vector Types and Linear Payoffs},
pdfauthor={Sébastien Huot, Abbas Edalat},
pdfkeywords={Multigame, Bayesian, Linear payoff, Multidimensional, Pure Bayesian Nash Equilibrium},
}

\begin{document}
\maketitle

\begin{abstract}
    We study $n$-agent Bayesian Games with $m$-dimensional vector types and linear payoffs, also called Linear Multidimensional Bayesian Games \citep{krishna}. This class of games is equivalent with $n$-agent, $m$-game Uniform Multigames introduced in \citep{multigames}. We distinguish between games that have a discrete type space and those with a continuous type space. More specifically, we are interested in the existence of pure Bayesian Nash Equilibrium for such games and efficient algorithms to find them. For continuous priors we suggest a methodology to perform Nash Equilibrium search in simple cases. For discrete priors we present algorithms that can handle two actions and two players games efficiently. We introduce the core concept of threshold strategy and, under some mild conditions, we show that these games have at least one pure Bayesian Nash Equilibrium. We illustrate our results with several examples like Double Game Prisoner Dilemna (DGPD), Chicken Game and Sustainable Adoption Decision Problem (SADP). 
\end{abstract}


Experimental results are accessible at this \href{https://github.com/huot-s/pure_ne_multigames}{repository}.

\section{Introduction}

Game theory provides an abstract framework to model a broad range of decision-making scenarios in real-life situations. From social cooperation \citep{cooperation_application} to biological evolution \citep{friedman1998} and economics \citep{shapiro1989} we can find models of games that can help to predict, or at least explain, the decision-makers' behaviour. \say{Game theory is a bag of analytical tools designed to help us understand the phenomena that we observe when decision-makers interact}\citep{osborne}. A \textit{game} is designed to model any situation involving \textit{decision-makers} (players, agents) that are \textit{rational} and reason \textit{strategically}. An agent is rational when it has a well-defined objective, is able to say which situation it prefers and
attempts to maximize its payoff. Strategic reasoning is to  use our
knowledge in the best possible way to anticipate the other agents’ action and take it into account in the decision process. The \textit{knowledge} of an agent is the information that it has
before making its decision. An information only known by one player is said to be
\textit{private knowledge}. When it is shared among a group of players, we say that this is
\textit{common knowledge}.\\

More formally, we assume agents are indexed by a set $I$. We denote the set of actions of agents by $A$, the set of outcomes by $O$ and consider an outcome function $f:A \rightarrow O$ giving the outcome of each action. The preferences of any agent are specified by the maximization of a utility (or payoff) function $u:O\rightarrow \R$. The simplest form of game is the \textit{Strategic Game} given in Definition~\ref{dfn:strategic_game}. An \textit{action profile} $a = (a_i)_{i \in I}$ is an outcome of the game. We use the index $-i$ to designate \say{all players except $i$} and write $a = (a_i, a_{-i})$ for $a \in A$.

\begin{dfn}\label{dfn:strategic_game}
A \textit{Strategic Game} $G = \left<I,(A_i)_{i \in I},(u_i)_{i \in I}\right>$ is made of:
\begin{enumerate}
    \item A finite set $I = \{ 1, 2, 3, ... , n\}$ of $n$ agents.
	\item For each agent $i \in I$ a nonempty set $A_i$ of possible actions. If $A_i$ is finite for all agents, then we say that the game is finite.
	\item A set $A = \prod_{i \in I} A_i$ of possible outcomes.
	\item A set of utility functions $u_i: A \rightarrow \R$ specifying the preferences of agent $i$.
\end{enumerate}
\end{dfn}

A \textit{Nash Equilibrium} (NE) \say{captures a steady state} \citep{osborne} of a game in which all agents have no incentive to deviate from their action unilaterally (Definition~\ref{dfn:nash_equ}).

\begin{dfn}\label{dfn:nash_equ}
A \textit{Nash Equilibrium} of the strategic game $G = \left<I,(A_i)_{i \in I},(u_i)_{i \in I}\right>$ is an action profile $a^* \in A$ such that:
\[ \forall i \in I \quad \forall a_i \in A_i \quad u_i(a_i^*,a_{-i}^*) \geq u_i(a_i,a_{-i}^*)\]
\end{dfn}
For a set $X$, let $\mathcal{P}(X)$ denote its set of subsets.
\begin{dfn}
Given an agent $i$ and the action $a_{-i} \in A_{-i}$ , the best-response function $B_i: A_{-i} \rightarrow \mathcal{P}(A_i)$ is such that:
\[ B_i(a_{-i}) = \{ a_i \in A_i \mid u_i(a_i,a_{-i}) \geq u_i(a_i',a_{-i}) \quad \forall a_i' \in A_i \} \]
\end{dfn}
This implies that, $a^*$ is a Nash Equilibrium if and only if $a^* \in B(a^*) := \left( B_i(a_{-i}) \right)_{i \in I}$. \\

The Prisoner Dilemma (PD) \citep{flood} is a widely used toy example that models a simple situation of self-interest driven decision and illustrates the case where the Nash Equilibrium is not efficient in the sense of Pareto \citep{osborne}. Extended to more than two agents, it's often associated with \textit{the tragedy of the commons} \citep{tradegy}. It is a situation in which \say{individuals, who have open access to a shared resource act independently according to their own self-interest and, contrary to the common good of all users, cause depletion of the resource through their uncoordinated action  \citep{tragedy_hardin}.} 

However, this behaviour may not be very common in real life as stated by Poundstone using the example of newspaper boxes in New Zealand \citep{poundstone}. Indeed, the PD can even fail to model real life examples.

Consider for example the TV show \say{Friend or Foe?} aired between 2002 and 2005 in the US. Two players who previously failed to answer some questions have to play \say{The Trust Box}: if they both choose Friend, they share a specified amount of money. If they both choose Foe, the money is lost. And if one plays Friend while the other plays Foe, the player choosing Foe wins all the money. Although (Friend, Friend) is not a NE, data collected from this kind of game showed that in such situation, \say{cooperation is surprisingly high} \citep{cooperative_behaviour}. Experience suggests that when a significant amount of money is at stake players tend to cooperate. Indeed, humans are not only driven by material considerations but also moral and sociocultural matters.

\subsection{Bayesian games} Concrete situations usually involve uncertainty at multiple levels so Strategic Games with perfect information (such as the Prisonner Dilemna or Hawk-Dove \citep{szekely2010}) may not be sufficient to describe the game. Usually we denote this unknown information by \textit{state of nature} and model it with a \textit{state space} $\Omega$ \citep{osborne}. Suppose each agent has some \textit{private information} (or \textit{private knowledge}): this information, called the \textit{type} $\theta_i$ of agent $i$, is unknown to the other agents, who can only have a (subjective) probability distribution over such possible types for an agent. So we assume that they have in mind a \textit{type space} $\Theta$ and a probability measure over $\Theta$, and, following \textit{von Neumann and Morgenstern}'s theory \citep{neumann_theory_of_games}, they play $a^*$ that maximizes the expected value of $u(f(a^*,\theta))$ subject to $\theta \in \Theta$. 

A game $G$ is \textit{Bayesian} \citep{harsanyi, harsanyi2, harsanyi3} when there is uncertainty for the agents (as in Definition~\ref{dfn:bayesian_game}). In this situation, an agent can no longer anticipate with certainty the output of the game because some information about the game is unknown. Variance of the outcome is generally associated to \textit{risk} or \say{how much we could diverge from what we expect}. A human player who wants to maximize his/her payoff without taking too much risk could be modeled by optimizing a combination of the expected value and the associated variance as seen in portfolio management models \citep{luenberg1998}. 

Bayesian Games also introduce the concept of \textit{strategy} which varies over the type of an agent. Indeed, the action of an agent depends on all the information it has (its own type) at the moment of the decision. We distinguish \textit{pure} strategies $s_i : \Theta_i \rightarrow A_i$ from \textit{mixed} strategies $\sigma_i : \Theta_i \rightarrow \Delta A_i$, where $\Delta X$ denotes the set of all probability distributions on a set $X$. In the latter case, the action played is not deterministic anymore but chosen randomly according to a probability distribution.\\

\begin{dfn}\label{dfn:bayesian_game}
A \textit{Bayesian Game} G is a game in strategic form with incomplete information which we denote by $G = \left<I, (A_i, \Theta_i, u_i, p_i(\cdot))_{i\in I}\right>$ where:
\begin{enumerate}
    \item $I = \{1, . . . , n\}$ is the set of agents.
	\item $A_i$ is agent $i$’s action set and $A = \prod_{i\in I} A_i$ is the set of action outcomes or action profiles.
	\item $\Theta_i$ is agent $i$’s type space and $\Theta = \prod_{i\in I} \Theta_i$ is the set of type profiles.
	\item $u_i : A \times \Theta \rightarrow \R$ is agent $i$’s utility for each $i\in I$.
	\item $p_i: \Theta \rightarrow \R$ is a (subjective) joint probability on $\Theta$ for each $i\in I$.
\end{enumerate}
\end{dfn}

For such games we can derive the notion of Bayesian Nash Equilibrium from the NE of strategic games where we maximize an expected utility $\overline{u_i}$ subject to agents' types. We also distinguish pure and mixed Nash Equilibria. Thanks to the Nash Theorem (Theorem~\ref{thm:nash_theorem}), we know that there always exists a mixed strategy Nash Equilibrium (denoted by \textit{mixed NE}). However, this theorem doesn't give a practical method to find this Nash Equilibrium. In fact, finding a Nash Equilibrium is often very complex and cannot be computed in a reasonable time. In 1994, Christos H. Papadimitriou introduced the complexity class PPAD (\textit{Polynomial Parity Arguments on Directed graphs}) \citep{complexity_NE} in order to classify the problem of finding a Nash Equilibrium. Later, he also showed that the problem of finding a Nash Equilibrium for a finite game is PPAD-complete \citep{complexity_NE_2}.\\

\begin{thm}[Nash Theorem \citep{osborne}]\label{thm:nash_theorem}
Every finite strategic game has a mixed strategy Nash equilibrium.
\end{thm}

Although finding a Nash Equilibrium is very hard in general, we can find some classes of games that have nice properties and an efficiently computable Nash Equilibrium \citep{rabinovich}.

\subsection{Linear multidimensional Bayesian games} 

For our study, we consider Linear Multidimensional Bayesian Games (Definition~\ref{dfn:linear_multidimensional_game}). They were initially introduced by Krishna and Perry to model multiple object auctions \citep{krishna}. This class of games was shown \citep{multigames} to be equivalent to another class of games called Uniform Multigames (Theorem~\ref{thm:equivalence_multigame}). A multigame (Definition~\ref{dfn:multigame}) is a game that is made of several local games that are played simultaneously by all agents. Each local game has its own payoff matrix and possible actions. For example an agent playing \say{Head or Tail} with one person and \say{Rock, Paper, Scissors} with another person simultaneously can be modeled by a multigame. Even if both local games are solved simultaneously, the agent is just playing on two \say{separate} games and tries to maximize a global utility which is a linear combination of local utilities. More specifically, for multigames in which every agent has the same set of actions in all local games, we can consider \textit{uniform} games: agents can only take the same action in all local games. When a multigame is uniform, an agent chooses only one strategy that is played identically in all local games. Said differently, one decision has multiple consequences and each agent wants to optimize their overall utility. We say that an agent chooses a global strategy that is applied on all the local games. Let $\R_+$ denote the set of non-negative real numbers.

\begin{dfn}[Linear Multidimensional Bayesian Games]\label{dfn:linear_multidimensional_game}
A Bayesian game $G$ is {\em $m$-dimensional} if the type space of each agent is a bounded subset of $\R^m_+$. When the positive integer $m > 1$ is implicitly given, we say $G$ is multidimensional. A multidimensional Bayesian game is linear if the utility of each agent only depends linearly on its own type components, i.e., there exists $L_i(s_i, s_{-i}) \in \R^m$ such that $u_i(s_i, s_{-i}, \theta_i, \theta_{-i}) = \sum_{j \in J} \left(L_i(s_i, s_{-i})\right)_j \theta_{ij}$.
\end{dfn}

\begin{dfn}\label{dfn:multigame}
A multigame  $$G =\langle I,J,(w_i)_{i \in I}, (G_j)_{j \in J},(\Theta_i)_{i \in I},(A_{ij}, u_{ij})_{i \in I, j \in J},p(\cdot)) \rangle$$ is a game in strategic form with incomplete information with the following structure:
\begin{enumerate}
    \item The set of agents $I=\{1,...,n\}$.
    \item The set of $n$-agent basic games is given by $G_j$, where $j \in J = \{1,...,m\}$ with action space $A_{ij}$ and utility function $u_{ij}$ for each agent $i \in I$ in the game $G_j$.
    \item Agent's $i$'s strategy is $s_i = (s_{ij})_{j \in J} \in S_i = \prod_{j \in J}A_{ij}$ where $s_{ij}$ is agent $i$'s action in $G_j$.
    \item Agent $i$'s type is $\theta_i = (\theta_{ij})_{j \in J} \in \Theta_i$ with $\theta_{ij} \geq 0$, $w_i > 0$ and $\sum_{j \in J}{\theta_{ij}} \leq w_i$
    \item Agent $i$'s utility for the strategy profile $(s_i,s_{-i})$ and type profile $(\theta_i, \theta_{-i})$ depends linearly on its types: $$u_i(s_i,s_{-i},\theta_i,\theta_{-i}) = \sum_{j \in J}{\theta_{ij}u_{ij}(s_{1j},...,s_{nj})}$$
    \item The agents' type profile $\theta = (\theta_1,...,\theta_n) \in \prod_{i \in I}\Theta_i$ is drawn from a given joint probability distribution $p(\theta)$. For any $\theta_i \in \Theta_i$, the function $p(\cdot \mid \theta_i)$ specifies a conditional probability distribution $\Theta_{-i}$ representing what agent $i$ believes about the types of the other agents if its own type where $\theta_i$.
\end{enumerate}
\end{dfn}

\begin{thm}\label{thm:equivalence_multigame}~\citep{multigames}
Suppose $G$ is a linear $m$-dimensional Bayesian game with a bounded type space $\Theta_i \subset \R^m_+$ for each $i \in I$, then $G$ is equivalent with a uniform multigame with $m$ basic games.
\end{thm}

Note that according to Definition~\ref{dfn:multigame}, the type can be any $m$-dimensional vector with positive coefficients. We say that a multigame is normalized when $\forall i \in I, w_i = 1$. Any multigame can be converted into a normalized multigame by adding a well-chosen local game \citep{multigames}. Thus we can assume without loss of generality that multigames are normalized and such that for all agents their type's coefficients add up to 1. As mentionned previously, this paper focuses on uniform multigames that have the following two basic features: (1) action sets are identical in all local games for every agent $i\in I$, i.e., $\forall j \in J, A_{ij} = A_i$ and (2) each agent plays the same action in all basic games $G_j$, i.e. $S_i = \{ (s,s,...,s) \mid s \in A_i \}$. Additionally, we make the assumption that the agents' types are independent $$\forall i \in I, \quad p_i(\theta_{-i} \mid \theta_i) = p_i(\theta_{-i})$$

We aim to show, beyond what is provided in~\citep{multigames}, that the multigame framework can model a wide range of complex situations that are worth exploring such as coordinated international environmental and social actions. First, take the situation in which two countries could either keep their traditional unsustainable production system or shift to a more responsible production. The shift is only beneficial if both countries do it. But each country can be tempted to keep the old traditional production system that doesn’t incur an additional cost and is more efficient and profitable (at least this is what agents trust). This situation falls into the original Prisoner Dilemma framework and we end up with both countries keeping their traditional production (assuming that they act rationally). 

Now, we extend this situation to $N$ companies in the same geographical area. We assume that each company that continues to keep an unsustainable production unit causes a  pollution cost of $1$ unit. Thus the (shared) pollution cost $k$ is the number of companies that don't shift their production unit, $0 \leq k \leq N$. On the other hand, a company that shifts its production unit for a more sustainable one has to pay an additional (fixed) cost $c \ll N$. Thus $c_i = k + c$ for a company $i$ that shifts its production and $c_i = k$ for a company that keeps it. If all $N$ agents shift, they all end up with a (small) cost $c$ and if no agent does so, they all end up with a (high) cost of $N$. Therefore, it is in the common interest that everyone shifts their production unit. But as each individual company has no incentive to deviate from not shifting, the only Nash Equilibrium of this game is that they all avoid the shift. The latter example can handle any number of companies but remains very limited as we cannot include the particularities of the agents (types) and the uncertainty about what other companies value the most (i.e., the Bayesian aspect). 

Later on, we study a more powerful approach through the multigames framework. We call it the \textit{Sustainable Adoption Decision Problem} (SADP): $n$ independent countries share $m$ possible concerns like population well-being, air pollution, economic stability, education or preserving biodiversity. Each country has its own (subjective) priorities characterized by an $m$-dimensional vector and has to choose between keeping their current lifestyle or radically shifting to a more sustainable one.

We denote by $G_{m,n,a}$ a \textit{uniform multigame} with $n$ agents, $m$ local games and $a$ number of possible actions assuming that all agents have the same number of possible actions. When we have games with 2 actions ($a=2$), we call these two actions C (cooperation) and D (defection). In this paper, we study $G_{m,n,2}$ multigames with continuous type space and show that with a simple condition on local games, the existence of a pure Bayesian Nash Equilibrium is guaranteed. At the same time, we define the notion of threshold strategy and discuss the possibility to extend it to any number of actions. Then, we operate the same kind of analysis for discrete games. For both continuous and discrete type spaces we define a particular kind of multigame $G_{2,2,2}$ that is called \textit{Double Game Prisoner Dilemma} (DGPD). In the last part, we propose algorithms that can efficiently find NE in particular situations and formulate postulates by exploring some properties of multigames.

\subsection{Related work}
Thanks to Nash Theorem, we know that any (finite) game has (at least) one Nash Equilibrium. But the number of equilibrium(s) and their nature is not described in general. There have been many researches to identify particular examples or sets of games that showcase patterns in the number or nature of their NE. For instance, \citep{EWERHART2021418} identifies that N-player Colonel Blotto games with incomplete information have Bayes–Nash equilibrium in which the resource allocation is strictly monotone. Probabilistic approaches \citep{random_DRESHER1970134, random_RINOTT2000274, random_STANFORD1995238, random_STANFORD1997115, random_STANFORD199929} have also been proposed, where the games are randomly generated under some constraints. The condition of \say{coarser inter-player information} introduced by \citep{HE201911} is shown to be necessary and sufficient to have pure Bayesian Nash Equilibrium. In \citep{EINY2023341}, the existence of pure-strategy equilibrium is proven for potential Bayesian games \citep{MONDERER1996124} with absolutely continuous information and a Bayesian potential that is upper semi-continuous in actions for any realization of the players' types. Similarly, we showcase a specific set of games that have pure Bayesian Nash Equilibrium and identify examples that illustrate our results.

\section{Uniform multigames with continuous type space}

\subsection{General remarks}

We use the standard Lebesgue measure on finite dimensional Euclidean spaces. When the type space is continuous, there are 3 possibilities: the probability distribution is either discrete, continuous or a mixture of both. If the probability distribution is fully discrete, we fall into the discrete type space case that we study later on. We choose to exclude the mixture case so that the distribution over the type space has no atomic value.\\

 We suppose that the probability distribution for the game is absolutely continuous with respect to the Lebesgue measure: Denote the Lebesgue measure on $\R^m$ by $\lambda_m$; if $p$ is a  probability distribution on $\R^m$, then the probability distribution $p$ is absolutely continuous with respect to $\lambda_m$ if for any measurable set $E$, we have: $\lambda_m(E) = 0 \implies p(E) = 0$. Let $p_i$ be the probability distribution for agent $i\in I$. We assume that we can use a probability density function: recall that, for any measurable set $E$, the probability density function $f_i$ satisfies: $$p_i(\theta_{-i} \in E) = \int_E f_i(\vec\theta_{-i})d\vec\theta_{-i}$$
In our case, $p_i(\vec\theta_{-i} \in E)$ is the probability, according to agent $i$, that the other agents have a type in the set $E \subset \Theta_{-i}$ and $f_i$ is the associated density function.\\

The type space $\Theta$ is assumed to be a compact subset of $\R^m$ and, thanks to the fact that the multigame is normalized, we have $dim(\Theta_i) = m-1$. \\

A pure strategy for agent $i$ is denoted by $s_i$, a mixed strategy by $\sigma_i$. We use $\sigma_i$ if we don't know a priori the nature of the strategy. We recall that for a pure strategy $s_i:\Theta_i \rightarrow A_i$ the agent plays an action for a given type. For a mixed strategy $\sigma_i : \Theta_i \rightarrow \Delta A_i$, the agent follows a probability for a given type where $\sigma_i(\vec\theta_i, a_i)$ is the probability of agent $i$ playing $a_i$ when their type is $\vec\theta_i$.\\

\subsection{Threshold strategy}

First, consider agent $i$'s expected utility $\overline{U}_i(a_i, \vec\theta_i, \sigma_{-i})$ given that it plays action $a_i$, has type $\vec\theta_i$ and other agents follow the strategy $\sigma_{-i}$:

\begin{align}
    \overline{U}_i(a_i, \vec\theta_i, \sigma_{-i}) = \int_{\vec\theta_{-i}}f_i(\vec\theta_{-i})U_i(a_i, \vec\theta_i, \sigma_{-i}(\vec\theta_{-i})) \,d\vec\theta_{-i}
\end{align}

$U_i(a_i, \vec\theta_i, \sigma_{-i}(\vec\theta_{-i}))$ is agent $i$'s utility playing $a_i$ with type $\vec\theta_i$ given that the others follow $\sigma_{-i}(\vec\theta_{-i})$. This utility can be expressed in terms of $U_i(a_i,\vec\theta_i, a_{-i})$ for $a_{-i} \in A_{-i}$:

\begin{align}
    U_i(a_i, \vec\theta_i, \sigma_{-i}(\vec\theta_{-i})) = \sum_{a_{-i}\in A_{-i}}U_i(a_i, \vec\theta_i, a_{-i})\sigma_{-i}(\vec\theta_{-i}, a_{-i})
\end{align}

As $G$ is a multigame, $U_i(a_i, \vec\theta_i, a_{-i})$ can be expressed as follows: 

\begin{align}
    U_i(a_i, \vec\theta_i, a_{-i}) = \sum_{j \in J}u_{ij}(a_i, a_{-i})\theta_{ij}
\end{align}

We define:

\begin{align}
    \zeta_i^{a_{-i}}(\sigma_{-i}) &:= \int_{\vec\theta_{-i}}f_i(\vec\theta_{-i})\sigma_{-i}(\vec\theta_{-i}, a_{-i})\,d\vec\theta_{-i} \\
    \overline{u}_{ij}(a_i, \sigma_{-i}) &:= \sum_{a_{-i}}u_{ij}(a_i, a_{-i})\zeta_i^{a_{-i}}(\sigma_{-i})
\end{align}

Here, $\zeta_i^{a_{-i}}(\sigma_{-i})$ is the probability that $a_{-i}$ is played by others given that they follow strategy $\sigma_{-i}$. Thus $\overline{u}_{ij}(a_i, \sigma_{-i})$ is the expected utility for agent $i$ in the local game $j$ if it plays action $a_i$ and others follow the strategy $\sigma_{-i}$. Using these, we can write:

\begin{eqnarray}
    \overline{U_i}(a_i, \vec\theta_i, \sigma_{-i}) &=& \int_{\vec\theta_{-i}}f_i(\vec\theta_{-i})\sum_{a_{-i}}\sigma_{-i}(\vec\theta_{-i}, a_{-i})\sum_{j \in J}u_{ij}(a_i, a_{-i})\theta_{ij}\,d\vec\theta_{-i}\\
    &=& \sum_{j \in J}\theta_{ij}\sum_{a_{-i}}u_{ij}(a_i, a_{-i})\int_{\vec\theta_{-i}}f_i(\vec\theta_{-i})\sigma_{-i}(\vec\theta_{-i}, a_{-i})\,d\vec\theta_{-i} \\
    &=& \sum_{j \in J}\theta_{ij}\sum_{a_{-i}}u_{ij}(a_i, a_{-i})\zeta_i^{a_{-i}}(\sigma_{-i}) \\
    &=& \sum_{j \in J}\theta_{ij}\overline{u}_{ij}(a_i, \sigma_{-i}) \\
    &=& \vec\theta_i \cdot (\overline{u}_{ij}(a_i, \sigma_{-i}))_{j \in J}
\end{eqnarray}

We thus have a more compact and explicit expression of the expected utility $\overline{U_i}(a_i, \vec\theta_i, \sigma_{-i})$. Indeed, it can be computed as the scalar product of the vector type of agent $i$ and the vector of the expected utilities for local games.\\

So far, we have kept game parameters $m, n$ and $a$ as general as possible, the previous expression holds for any choice of those parameters. To go further on with the analysis, we consider that $a = 2$. We will discuss in the conclusion section the extension to any number of actions.\\

We aim to evaluate whether agent $i$ with type $\vec\theta_i$ and opponents' strategy $\sigma_{-i}$ prefers to play $C$ or $D$. We compare $\overline{U_i}(C, \vec\theta_i, \sigma_{-i})$ and $\overline{U_i}(D, \vec\theta_i, \sigma_{-i})$:

\begin{eqnarray}
    \overline{U_i}(C, \vec\theta_i, \sigma_{-i}) - \overline{U_i}(D, \vec\theta_i, \sigma_{-i}) &=& \sum_{j \in J}\theta_{ij}(\overline{u}_{ij}(C, \sigma_{-i}) - \overline{u}_{ij}(D, \sigma_{-i})) \\
    &=& \sum_{j \in J}\theta_{ij}\delta_{ij}(\sigma_{-i}) \\
    &=& \vec\theta_i \cdot \vec\delta_i(\sigma_{-i})
\end{eqnarray}
where $\delta_{ij}(\sigma_{-i}) = \overline{u}_{ij}(C, \sigma_{-i}) - \overline{u}_{ij}(D, \sigma_{-i})$ and $\vec \delta_i(\sigma_{-i}) = (\delta_{ij}(\sigma_{-i}))_{j \in J}$. This difference expressed as a scalar product indicates the best action for agent $i$: if it is strictly positive then the best action is $C$, if it is strictly negative then the best action is $D$ and if it is equal to zero then any mixed combination of $C$ and $D$ is a best response.

\begin{dfn}[Threshold strategy]\label{TS}
    The vector $\vec\delta_i$ is called agent $i$'s {\em threshold}. A {\em threshold strategy} $\sigma_i$ with threshold $\vec\delta_i$ for agent $i$ is a strategy such that:
    \begin{align}
        \sigma_i(\vec\theta_i) = 
        \begin{cases}
            C &\text{$\vec \theta_i \cdot \vec \delta_i > 0$} \\
            D &\text{$\vec \theta_i \cdot \vec \delta_i < 0$} \\
            \alpha_i(\vec \theta_i) C + (1 - \alpha_i(\vec \theta_i)) D &\text{$\vec \theta_i \cdot \vec \delta_i = 0$},
        \end{cases} 
    \end{align}
    where $\alpha_i(\vec\theta_i)\in [0,1]$. Such a strategy is also denoted by $(\sigma_i, \vec \delta_i, \alpha_i)$.
\end{dfn}

 The first and second cases ($\vec \theta_i \cdot \vec \delta_i \neq 0$) are called \textit{pure components} of the strategy and the last case ($\vec \theta_i \cdot \vec \delta_i = 0$) is called the \textit{mixed component}. Notice that, as a direct consequence of what we previously said, a best response is always a threshold strategy and thus any Bayesian Nash Equilibrium $\sigma=(\sigma_1, \sigma_2, ..., \sigma_n)$ is exclusively made of threshold strategies. Also, a threshold strategy is said to be {\em pure} when $\forall \vec\theta_i$, $\alpha_i(\vec\theta_i) \in \{0, 1\}$.

\subsection{Existence of pure Bayesian Nash Equilibrium}

\begin{thm}\label{thm:main_theorem}
If a (normalized) uniform multigame $G_{m,n,2}$ with continuous type space and continuous prior has a mixed strategy Nash Equilibrium $(\sigma_1, \sigma_2, ..., \sigma_n)$ with non-zero threshold vectors for all agents then it has a pure Bayesian Nash Equilibrium.
\end{thm}

\begin{proof}
 Consider a mixed Nash Equilibrium $\sigma=(\sigma_1, \sigma_2, ..., \sigma_n)$ that satisfies the condition given in the theorem. We show that starting from the mixed strategy $\sigma$ we can derive $s=(s_1, s_2, ..., s_n)$ a pure strategy Nash Equilibrium.

First, we notice that $\sigma_i$'s are threshold strategies $(\sigma_i, \vec \delta_i^*, \alpha_i)$ with threshold $\vec \delta_i^* = \vec \delta_i(\sigma_{-i})$ as they are best responses for each agent given the opponents' strategies. We suppose that there is at least one agent $i$ such that $\sigma_i$ is a mixed strategy ($\alpha_i \neq 0,1$) because otherwise $\sigma$ is already a pure strategy Nash Equilibrium and the proof is over. For each $i\in I$, derive the pure strategy $s_i$ from $\sigma_i$ by replacing the mixed component with a pure action, i.e., by setting $\alpha_i \in \{0, 1\}$. We now demonstrate that $s=(s_1, s_2, ..., s_n)$ is a Bayesian Nash Equilibrium for $G$.

By construction, $s_i$'s are threshold strategies with the same threshold as $\sigma_i$'s. So if we can show that $\vec \delta_i(\sigma_{-i}) = \vec \delta_i(s_{-i})$ for all agents $i$ then $s$ is a Bayesian Nash Equilibrium. For this purpose, we prove that: $$\forall i \in I, \forall a_{-i} \in A_{-i}. \quad \zeta^{a_{-i}}_i(\sigma_{-i}) = \zeta^{a_{-i}}_i(s_{-i})$$ by computing the difference:
\begin{align}
    \zeta^{a_{-i}}_i(\sigma_{-i}) - \zeta^{a_{-i}}_i(s_{-i}) = \int_{\vec\theta_{-i}}f_{-i}(\vec\theta_{-i})[\sigma_{-i}(\vec\theta_{-i}, a_{-i}) - s_{-i}(\vec\theta_{-i}, a_{-i})]\,d\vec\theta_{-i}.
\end{align}

Since $\vec \delta_i(\sigma_{-i})$ is a  non zero vector, the set $E_i = \{ \vec\theta_i \mid \vec\theta_i \cdot \vec \delta_i(\sigma_{-i}) = 0 \}$ is contained in a  hyperplane of $\R^{m-1}$. Thus, for any agent $i$, the set $\Theta_1 \times ... \times E_i \times ... \times \Theta_n$ has zero measure. Also, note that $\sigma_{-i}(a_{-i}, \vec\theta_{-i}) \neq s_{-i}(a_{-i}, \vec\theta_{-i})$ only when $\vec\theta_{-i} \in \bigcup_{k \neq i}\Theta_{-\{i,k\}} \times E_k$ where the latter set is the finite union of nullsets. In other words, $\sigma_{-i}(\vec\theta_{-i}, a_{-i})$ and $s_{-i}(\vec\theta_{-i}, a_{-i})$ are equal almost everywhere. So the difference expressed by the integral is zero. This shows that the constructed pure strategy $s$ is a Nash Equilibrium solution.\\
\end{proof}

Theorem \ref{thm:main_theorem} provides a pure Bayesian Nash Equilibrium for uniform multigames with 2 actions but it relies on a specific condition on mixed strategy solutions. In practice, we do not seek to enumerate all possible mixed solutions just to check whether any one of them has non-zero threshold vectors for all agents. Fortunately, we can find conditions that don't rely on the mixed solutions but can help us to determine whether there exists pure solutions.

\begin{dfn}[Purely cooperative/competitive local game]
    A local game $j \in J$ is said to be {\em purely cooperative} for agent $i \in I$ if $$\forall a_{-i}\in A_{-i}. \quad u_{ij}(C, a_{-i}) > u_{ij}(D, a_{-i})$$ and to be {\em purely competitive} if $$\forall a_{-i}\in A_{-i}. \quad u_{ij}(C, a_{-i}) < u_{ij}(D, a_{-i})$$
    
\end{dfn}

This notion is equivalent to the condition that C (or D) is a strictly dominant strategy for agent $i$ in game $j$.

\begin{prop}\label{prop:purely_competitive}
    An agent having at least one purely competitive (or cooperative) local game will always play a threshold strategy with a non-zero threshold, whatever the opponents' strategy $\sigma_{-i}$.
\end{prop}

\begin{proof}
    Assume that game $k$ is purely cooperative for agent $i$ (the same reasoning applies to a purely competitive game). Let's compute $\delta_{ik}(\sigma_ {-i})$:
    
    \begin{align}
        \delta_{ik}(\sigma_{-i}) = \sum_{a_{-i}}(u_{ik}(C, a_{-i}) - u_{ik}(D, a_{-i}))\zeta_i^{a_{-i}}(\sigma_{-i}).
    \end{align}
    Since the terms $\zeta_i^{a_{-i}}(\sigma_{-i})$  express probabilities, they are non-negative with $\sum_{a_{-i}}\zeta_i^{a_{-i}}(\sigma_{-i}) = 1$. Because the local game $k$ is purely cooperative for agent $i$, we have $u_{ik}(C, a_{-i}) - u_{ik}(D, a_{-i})>0$ for all $k\in J$. Thus the sum is strictly positive and $\delta_{ik}(\sigma_{-i}) > 0 $ which implies $\vec \delta_i(\sigma_{-i}) \neq 0$.

\end{proof}

In the light of Proposition~\ref{prop:purely_competitive} we easily deduce the following:

\begin{thm}\label{thm:practical_continuous_theorem}
Any normalized uniform multigame $G_{m,n,2}$ with continuous type space and continuous prior, with at least one purely competitive/cooperative local game for each agent, has a pure Bayesian Nash Equilibrium.
\end{thm}

\begin{proof}
    We note that, by Proposition~\ref{prop:purely_competitive}, all agents must play a threshold strategy with a non-zero threshold. So any mixed Nash Equilibrium must be made of threshold strategies with non-zero thresholds. Thus the conditions of Theorem~\ref{thm:main_theorem} hold, and a pure Bayesian Nash Equilibrium exists.
\end{proof}

\subsection{The Double Game Prisoner Dilemma}

In this subsection, we define and study a more particular type of multigame from \citep{pd_extension}  called Double Game Prisoner Dilemma (DGPD). On the discrete type space section we will also refer to such multigames.

\begin{dfn}
    The Double Game Prisoner Dilemma is a $(2,2,2)$-multigame (2 agents, 2 local games, 2 actions) such that the first local game is a Prisoner's Dilemma game and the second local game is a \say{social game} motivating cooperation with the following payoff matrices:
    
    \begin{table}[ht]
        \begin{minipage}{.5\linewidth}
            \centering
            \begin{tabular}{cc|c|c|}
                \multicolumn{2}{c}{} & \multicolumn{2}{c}{agent 2} \\
                \cline{3-4}
                & & $C$ & $D$\\ \cline{2-4}
                \multicolumn{1}{c}{\multirow{2}{*}{agent 1} } &
                \multicolumn{1}{|c|}{$C$} & $(r,r)$ & $(s,t)$\\ \cline{2-4}
                \multicolumn{1}{c}{}                        &
                \multicolumn{1}{ |c| }{$D$} & $(t,s)$ & $(p,p)$\\ \cline{2-4}
            \end{tabular}
            \caption{Prisoner's Dilemma Payoffs}
        \end{minipage}
        \begin{minipage}{.5\linewidth}
            \centering
            \begin{tabular}{cc|c|c|}
                \multicolumn{2}{c}{} & \multicolumn{2}{c}{agent 2} \\
                \cline{3-4}
                & & $C$ & $D$\\ \cline{2-4}
                \multicolumn{1}{c}{\multirow{2}{*}{agent 1} } &
                \multicolumn{1}{|c|}{$C$} & $(y,y)$ & $(y,z)$\\ \cline{2-4}
                \multicolumn{1}{c}{}                        &
                \multicolumn{1}{ |c| }{$D$} & $(z,y)$ & $(z,z)$\\ \cline{2-4}
            \end{tabular}
            \caption{Social Game Payoffs}
        \end{minipage} 
    \end{table}
    
    and such that $$s = z, \quad t > r > y > p > s, \quad r > (t + s)/2, \quad  y > (r + p)/2$$

    There are 2 games, so we choose to denote the vector types of both agents with $\begin{bmatrix}1 - \theta_1 \\ \theta_1\end{bmatrix}$ and $\begin{bmatrix}1 - \theta_2 \\ \theta_2\end{bmatrix}$, $\theta_i$ is called the \textit{pro-social coefficient} of agent $i$.
\end{dfn}

\begin{figure}[ht]
\centering
\includegraphics[width = \hsize]{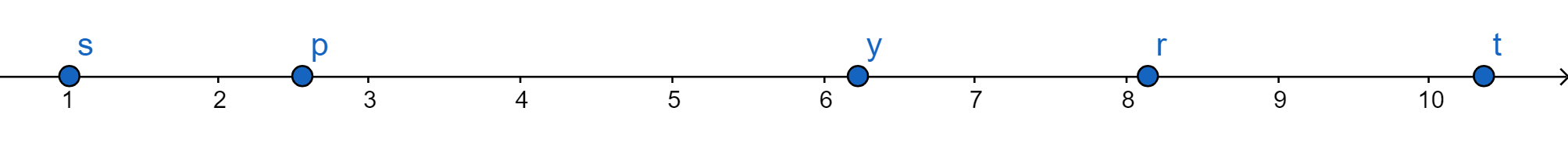}
\caption{Example of coefficients verifying the DGPD conditions}
\end{figure}

As a direct consequence of Theorem \ref{thm:practical_continuous_theorem}, we can easily deduce the existence of pure Bayesian Nash Equilibrium.

\begin{corollary}
For any DGPD with continuous type space and continuous prior, there exists a pure Bayesian Nash Equilibrium (made of threshold strategies).
\end{corollary}

\begin{proof}
    This is a direct consequence of Theorem \ref{thm:practical_continuous_theorem} as the social game is a purely cooperative game for both agents.
\end{proof}

Replacing $u_{ij}(a_i, a_{-i})$ and $\vec\theta_i$ in the previous computations by DGPD parameters gives us: 

\begin{align}
    \overline{U_i}(C, \theta_i, \sigma_{-i}) &= (1 - \theta_i)\left[ \zeta_i^C(\sigma_{-i}) r + \zeta_i^D(\sigma_{-i}) s\right] + \theta_i y\\
    \overline{U_i}(D, \theta_i, \sigma_{-i}) &= (1 - \theta_i)\left[ \zeta_i^C(\sigma_{-i}) t + \zeta_i^D(\sigma_{-i}) p\right] + \theta_i s
\end{align}

Note that for a given $\sigma_{-i}$, the expected values $\overline{U_i}(C, \theta_i, \sigma_{-i})$ and $\overline{U_i}(D, \vec\theta_i, \sigma_{-i})$ are linear functions in $\theta_i$ and they cross at $\theta_i = \theta_i^* \in [0,1]$ since we have: $$y \geq s, \quad \zeta_i^C(\sigma_{-i}) r + \zeta_i^D(\sigma_{-i}) s \leq \zeta_i^C(\sigma_{-i}) t + \zeta_i^D(\sigma_{-i}) p$$ 

As a result of this, we can formulate a more convenient definition for threshold strategy in the DGPD context.

\begin{dfn}[DGPD Threshold strategy]
    A threshold strategy $\sigma_i$ with threshold $\theta_i^*$ for agent $i$ is a strategy such that:
    \begin{align}
        \sigma_i(\theta_i) = 
        \begin{cases}
            D & \text{$\theta_i < \theta_i^*$}\\
            C & \text{$\theta_i > \theta_i^*$}\\
            \alpha_i C + (1 - \alpha_i) D & \text{$\theta_i = \theta_i^*$}
        \end{cases}
    \end{align}
\end{dfn}

Again, by construction, a best response must be a threshold strategy as we have just rearranged the notation compared to the general definition. Note that when the pro-social coefficient is low (i.e., below the threshold) agent $i$ defects, and when the pro-social coefficient is high (i.e., above the threshold) agent $i$ cooperates. Figure \ref{fig:conclusion_1} summarizes the concept of threshold strategy in the context of DGPD.\\

\begin{figure}[ht]
\centering
\includegraphics[width = 0.95\hsize]{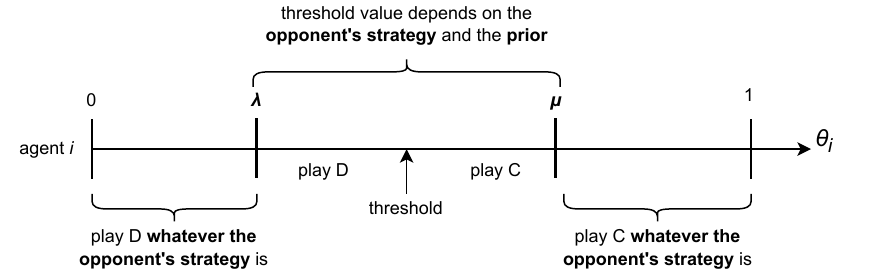}
\caption{DGPD strategy summarized ($\lambda < \mu$).}
\label{fig:conclusion_1}
\end{figure}

We define the {\em threshold function} $\theta_i^*:[0,1] \rightarrow \R$ by:

\begin{align}\label{eqn:threshold_func}
   \theta_i^*(x) &= \frac{x(t - r) + (1 - x)(p - s)}{x(t - r) + (1 - x)(p - s) + (y - s)} 
\end{align}

By a simple calculation, we conclude that, given the other agent strategy $\sigma_{-i}$, agent $i$'s best response is a threshold strategy with threshold $\theta_i^* = \theta_i^*(\zeta_i^C(\sigma_{-i}))$.

We now define $\lambda$ and $\mu$ that are combinations of DGPD payoff parameters as in~\citep{multigames}:
\begin{align}
    \mu := \theta_i^*(0) = \frac{p-s}{p-s+y-s}, \quad \quad
    \lambda := \theta_i^*(1) = \frac{t-r}{t-r+y-s}
\end{align}

\begin{prop}[$\theta_i^*$ monotonicity] \label{prop:threshold_monotonicity}
    The threshold function $\theta_i^*:[0,1] \rightarrow \R$ is monotonic and is increasing if $\mu<\lambda$, decreasing if $\mu>\lambda$ and is constant for $\mu=\lambda$. 
 \end{prop}   
\begin{proof}
   The monotonicity is straightforward. Then we just note that $\theta_i^*(0) = \mu$ and $\theta_i^*(1) = \lambda$.
\end{proof}

\begin{prop}
    For any DGPD with continuous type space and continuous prior, there exists a pure Bayesian Nash Equilibrium. This equilibrium is made of threshold strategies with thresholds $\theta_i^* \in [\min (\lambda, \mu), \max (\lambda, \mu)]$. 
\end{prop}

\begin{proof}
    The predicate $\theta_i^* =  \theta_i^*(\zeta_i^C(\sigma_{-i})) \in [\min (\lambda, \mu), \max (\lambda, \mu)]$ results from the fact that $\zeta_i^C(\sigma_{-i}) \in [0, 1]$ and $\theta_i^*(\zeta_i^C(x))$ is monotonic in $[0, 1]$.
\end{proof}

\subsection{Example for the SADP}

Consider a simple situation with two companies, each having two different visions and two possible actions. Assume, for example, that the first vision is to protect the brand reputation and the second is to gain market shares. Suppose that there's an opportunity to establish a new production facility in a controversial area. Seizing the opportunity will for sure help gain market shares but will also impact negatively on the reputation. Each company can either leave this opportunity (C) or compete to set up a new site in the area (D). The following tables summarize this situation:

\begin{table}[ht]
    \begin{minipage}{.5\linewidth}
        \centering
        \begin{tabular}{cc|c|c|}
            \multicolumn{2}{c}{} & \multicolumn{2}{c}{company B} \\
            \cline{3-4}
            & & $C$ & $D$\\ \cline{2-4}
            \multicolumn{1}{c}{\multirow{2}{*}{company A} } &
            \multicolumn{1}{|c|}{$C$} & $(5,5)$ & $(0,9)$\\ \cline{2-4}
            \multicolumn{1}{c}{}                        &
            \multicolumn{1}{ |c| }{$D$} & $(9,0)$ & $(2,2)$\\ \cline{2-4}
        \end{tabular}
        \caption{Market share Game}
    \end{minipage}
    \begin{minipage}{.5\linewidth}
        \centering
        \begin{tabular}{cc|c|c|}
            \multicolumn{2}{c}{} & \multicolumn{2}{c}{company B} \\
            \cline{3-4}
            & & $C$ & $D$\\ \cline{2-4}
            \multicolumn{1}{c}{\multirow{2}{*}{company A} } &
            \multicolumn{1}{|c|}{$C$} & $(4,4)$ & $(4,0)$\\ \cline{2-4}
            \multicolumn{1}{c}{}                        &
            \multicolumn{1}{ |c| }{$D$} & $(0,4)$ & $(0,0)$\\ \cline{2-4}
        \end{tabular}
        \caption{Reputation Game}
    \end{minipage} 
\end{table}

Both companies have a continuous type space and prior. According to our last result, there exists a pure Bayesian Nash Equilibrium made of threshold strategies. By computing $\lambda = 1/2$ and $\mu = 1/3$, we also know that the thresholds $\theta_i^*$ are in $[\frac{1}{3},\frac{1}{2}]$. If we add the assumption that priors are uniform, we can show (see Section~\ref{algRes_cont} on Algorithmic results) that $$\theta_1^* = \theta_2^* = \frac{5 - \sqrt{17}}{2} \approx 0.4384$$ 

\subsection{Application to Other 2-Players Games}

In this section we present other examples of multigames for which we can apply our result on the existence of pure NE.

\subsubsection{A Chicken Game variation}
In the Chicken Game (also called Hawk-Dove Game) \citep{osborne} there are two agents that can either go for the conflict (Conflict) or avoid it (Avoid). The best outcome for an agent is to play Conflict while the other plays Avoid. As opposed to the PD, if they both play Conflict, they both face the worst outcome. An example is given in Table~\ref{tab:chicken_payoff}.

\begin{table}[h!]
    \centering    
	\begin{tabular}{|c|c|c|}
		\cline{2-3} 
		\multicolumn{1}{c|}{} & Avoid & Conflict  \\ \hline
		Avoid & $(2, 2) $&$(1, 3)$ \\\hline
		Conflict &$(3, 1) $&$(0, 0)$\\\hline
	\end{tabular}
	\caption{The Chicken Game payoff matrix}
	\label{tab:chicken_payoff}
\end{table}

Because it is mainly a toy example like the PD, this game can only model very specific situations. We suggest that the game can be understood as a combination of two drivers: ego and survival. Under the ego consideration, we want to play Conflict not to be considered the ``chicken". What's important is to show that we dominate our opponent and have a stronger mind. Under the survival consideration we mainly want to avoid the situation where both agents play Conflict. An example of both games is given in Table~\ref{double_chicken_payoff_ego} and \ref{double_chicken_payoff_survival}. The double game made of those two games with continuous priors has a pure Nash Equilibrium because the ego game has a strictly dominant strategy (Conflict).

\begin{table}[ht]
    \begin{minipage}{.5\linewidth}
        \centering
        \begin{tabular}{cc|c|c|}
            \multicolumn{2}{c}{} & \multicolumn{2}{c}{agent 2} \\
            \cline{3-4}
            & & Avoid & Conflict\\ \cline{2-4}
            \multicolumn{1}{c}{\multirow{2}{*}{agent 1} } &
            \multicolumn{1}{|c|}{Avoid} & $(0,0)$ & $(-1,1)$\\ \cline{2-4}
            \multicolumn{1}{c}{}                        &
            \multicolumn{1}{ |c| }{Conflict} & $(1,-1)$ & $(0,0)$\\ \cline{2-4}
        \end{tabular}
        \caption{Ego Game}
        \label{double_chicken_payoff_ego}
    \end{minipage}
    \begin{minipage}{.5\linewidth}
        \centering
        \begin{tabular}{cc|c|c|}
            \multicolumn{2}{c}{} & \multicolumn{2}{c}{agent 2} \\
            \cline{3-4}
            & & Avoid & Conflict\\ \cline{2-4}
            \multicolumn{1}{c}{\multirow{2}{*}{agent 1} } &
            \multicolumn{1}{|c|}{Avoid} & $(0,0)$ & $(0,0)$\\ \cline{2-4}
            \multicolumn{1}{c}{}                        &
            \multicolumn{1}{ |c| }{Conflict} & $(0,0)$ & $(-2,-2)$\\ \cline{2-4}
        \end{tabular}
        \caption{Survival Game}
        \label{double_chicken_payoff_survival}
    \end{minipage} 
\end{table}

\subsubsection{A Battle of Sexes Variation}

In the Battle of Sexes (also called Bach or Stravinsky), there are two agents that want to meet in a event but have opposed tastes. The main goal of both is to spend time together but they also value the fact of going to the event they like the most.  

\begin{table}[h!]
    \centering
	\begin{tabular}{|c|c|c|}
		\cline{2-3} 
		\multicolumn{1}{c|}{} & Bach & Stravinsky  \\ \hline
		Bach & $(10, 7) $&$(2, 2)$ \\\hline
		Stravinsky &$(0, 0) $&$(7, 10)$\\\hline
	\end{tabular}
	\caption{The Bach or Stravinsky payoff matrix}
	\label{tab:battle_sexes_payoff}
\end{table}

As presented in the previous example, we can try to decompose considerations into Taste and Social. In the Taste Game (Table~\ref{double_battle_sexes_payoff_taste}) both agents are motivated to follow their taste no matter what the other does. In the Social Game (Table~\ref{double_battle_sexes_payoff_social}), the agents want to be at the same event, no matter which event it is.

\begin{table}[ht]
    \begin{minipage}{.5\linewidth}
        \centering
        \begin{tabular}{cc|c|c|}
            \multicolumn{2}{c}{} & \multicolumn{2}{c}{agent 2} \\
            \cline{3-4}
            & & B & S\\ \cline{2-4}
            \multicolumn{1}{c}{\multirow{2}{*}{agent 1} } &
            \multicolumn{1}{|c|}{B} & $(1,0)$ & $(1,1)$\\ \cline{2-4}
            \multicolumn{1}{c}{}                        &
            \multicolumn{1}{ |c| }{S} & $(0,0)$ & $(0,1)$\\ \cline{2-4}
        \end{tabular}
        \caption{Taste Game}
        \label{double_battle_sexes_payoff_taste}
    \end{minipage}
    \begin{minipage}{.5\linewidth}
        \centering
        \begin{tabular}{cc|c|c|}
            \multicolumn{2}{c}{} & \multicolumn{2}{c}{agent 2} \\
            \cline{3-4}
            & & B & S\\ \cline{2-4}
            \multicolumn{1}{c}{\multirow{2}{*}{agent 1} } &
            \multicolumn{1}{|c|}{B} & $(2,2)$ & $(0,0)$\\ \cline{2-4}
            \multicolumn{1}{c}{}                        &
            \multicolumn{1}{ |c| }{S} & $(0,0)$ & $(2,2)$\\ \cline{2-4}
        \end{tabular}
        \caption{Social Game}
        \label{double_battle_sexes_payoff_social}
    \end{minipage} 
\end{table}

And like before, the double game made of those two games with continuous priors has a pure Nash Equilibrium because the Taste Game has a strictly dominant strategy (Bach for agent 1 and Stravinsky for agent 2).

\subsubsection{An Assurance Game Variation}

The Assurance Game (or Stag Hunt) is a 2-players game involving a conflict between personal safety and social cooperation. An example of payoff matrix is given in Table~\ref{tab:stag_hunt_payoff}. By essence, this game is very similar to Prisoner Dilemma. Thus we naturally combine it with a social game and get a variation of the DGPD.

\begin{table}[h!]
    \centering
	\begin{tabular}{|c|c|c|}
		\cline{2-3} 
		\multicolumn{1}{c|}{} & Stag & Hunt  \\ \hline
		Stag & $(10, 10) $&$(1, 8)$ \\\hline
		Hunt &$(8, 1) $&$(5, 5)$\\\hline
	\end{tabular}
	\caption{The Stag Hunt payoff matrix}
	\label{tab:stag_hunt_payoff}
\end{table}

\subsection{Algorithmic results}\label{algRes_cont}
\subsubsection{Uniform prior}

Recall the definition of $\zeta_i^C(\sigma_{-i})$ and $\theta_i^*(\zeta_i^C)$ that characterize the best response of an agent given its opponent's strategy (assuming that the both follow a threshold strategy): 

\begin{align*}
    \zeta_i^C(\sigma_{-i}) = p_{-i}(\theta_{-i} \in [\max (\lambda, \mu), 1]) + \int_{\min (\lambda, \mu)}^{\max (\lambda, \mu)} f_{-i}(\theta_{-i})\sigma_{-i}(\theta_{-i}, C)d\theta_{-i}
\end{align*}

\begin{align*}
    \theta_i^*(x) = \frac{x(t - r) + (1 - x)(p - s)}{x(t - r) + (1 - x)(p - s) + (y - s)}
\end{align*}

In the case of a uniform prior, $p_{-i}(\theta_{-i} \in [\max (\lambda, \mu), 1]) = 1 - \max (\lambda, \mu)$. For simplicity, we use the notation $\alpha = \min (\lambda, \mu)$ and $\beta = \max (\lambda, \mu)$. Given that the opponent plays a threshold strategy with the threshold $\theta_{-i}^*$, we have:

\begin{align*}
    \zeta_i^C(\sigma_{-i}) &= 1 - \beta + \int_{\theta_{-i}^*}^{\beta}f_{-i}(\theta_{-i})d\theta_{-i} \\
    &= 1 - \beta + (\beta - \theta_i^*) \\
    &= 1 - \theta_i^*
\end{align*}

Therefore, we can rewrite the threshold function $\theta_i^*$ as a function of $\theta_{-i}^*$:

\begin{align*}
    \theta_i^*(\theta_{-i}^*) &= \frac{(1 - \theta_{-i}^*)(t - r + s - p) + (p - s)}{(1 - \theta_{-i}^*)(t - r + s - p) + (p - s + y - s)} \\
    &= \frac{d(1 - \theta_{-i}^*) + e}{d(1 - \theta_{-i}^*) + e + f},
\end{align*}
where $$d := t - r + s - p, \quad e := p - s, \quad f := y - s$$ \newline

First, assume that the solution is symmetric for both players, meaning that $\theta_1^* = \theta_2^*$:
\begin{align*}
    \theta_i^* = \frac{d(1 - \theta_i^*) + e}{d(1 - \theta_i^*) + e + f}
\end{align*}
which is reduced to:
\begin{align}
    -d(\theta_i^*)^2 + (2d + e + f)\theta_i^* - (d + e) = 0
\end{align}
This is a quadratic equation. To evaluate the number of solutions we compute the discriminant:

\begin{align}\label{quadratic_equ}
    \Delta &= (2d + e + f)^2 - 4d(d + e)\\
    &= (e + f)^2 + 4df
\end{align}

If we search for non-symmetric solutions, the condition $\theta_i^* = \theta_{-i}^*$ doesn't hold anymore. By doing the same kind of computation we end up with:

\begin{align}\label{quadradtic_equ_2}
    -d(e + f)(\theta_i^*)^2 + ((d + e + f)^2 - d^2)\theta_i^* - (e(d+e+f) + df) = 0
\end{align}

This is also a quadratic equation and we notice that it is the same quadratic as the symmetric case but is multiplied by the constant $(e + f) \neq 0$. Hence, both equations have exactly the same solutions.

\begin{prop}
    Under the DGPD assumptions, the quadratic equations~(\ref{quadratic_equ}) and (\ref{quadradtic_equ_2}) always have 2 solutions.
\end{prop}

\begin{proof}
    The discriminant of those quadratics can be written as $\Delta = (e - f)^2 + 4(d + e)f$. By the constraints on the DGPD parameters, $f > 0$ and $d+e = t - r > 0$, and thus we have $\Delta > 0$ and the quadratics have two solutions $x_- = \frac{-(2d + e + f) - \sqrt{\Delta}}{-2d}$ and $x_+ = \frac{-(2d + e + f) + \sqrt{\Delta}}{-2d}$.
\end{proof}

Now, we need to check the validity of the solutions, i.e. $\alpha \leq x_{sol} \leq \beta$.

\begin{prop}
    Among the two solutions of the quadratics (\ref{quadratic_equ}) and (\ref{quadradtic_equ_2}), $x_{+}$ is always valid and $x_{-}$ is always invalid.
\end{prop}

\begin{proof}
    By computing $x_{-} - \mu$ and $x_{-} - \lambda$ we can notice that depending on the sign of $d$, we have either $x_{-} > \max(\mu, \lambda)$ or $x_{-} < \min(\mu, \lambda)$. In both case the solution is not valid. With the same kind of reasoning, we can show that depending on the sign of $d$, we have either $\mu < x_{+} < \lambda$ or $\lambda <x_{+} < \mu$. Thus $x_{+}$ is always valid.
\end{proof}

We end up with the fact that such games have always exactly one pure NE which is symmetric (i.e. both players have the same threshold) where $\theta_1^* = \theta_2^* = x_{+}$.

\subsubsection{General solution}

We denote by $F_i(x)$ the cumulative distribution function of agent $i$'s prior $f_{-i}(\theta_{-i})$. In the case of uniform prior, we were able to explicitly find $F_i(\cdot)$ and deduce an algebraic equation for the solutions. To find a solution for any prior, we have to solve the follwing set of equations:
\begin{align}
    \theta_1^* &= \frac{A_1 - dF_1(\theta_2^*)}{B_1 - dF_1(\theta_2^*)} \\
    \theta_2^* &= \frac{A_2 - dF_2(\theta_1^*)}{B_2 - dF_2(\theta_1^*)}
\end{align}
with
\begin{align*}
    \zeta_i^C(\sigma_{-i}) &= F_i(1) - F_i(\theta_{-i}^*) \\
    A_i &= dF_i(1) + e\\
    B_i &= A_i + f
\end{align*}

Therefore, for continuous type space, the NE search is equivalent to solving a nonlinear multivariate equation:

\begin{align}
    f(\theta_1^*, \theta_2^*) = \left(\theta_1^* - \frac{A_1 - dF_1(\theta_2^*)}{B_1 - dF_1(\theta_2^*)}, \theta_2^* - \frac{A_2 - dF_2(\theta_1^*)}{B_2 - dF_2(\theta_1^*)}\right)\\
    \text{find $(\theta_1^*, \theta_2^*)$ such that $f(\theta_1^*, \theta_2^*) = 0$}
\end{align}

\section{Uniform multigames with discrete type space}
\subsection{General Remarks}

The type space is now assumed to be discrete. The term $\overline{U}_i(a_i, \vec\theta_i, \sigma_{-i})$, for $i\in I$, is the same as the continuous prior except that the integral is now replaced with a sum. Thus we keep the same notations and consider: 
\begin{align}
    \zeta_i^{a_{-i}}(\sigma_{-i}) := \sum_{\vec\theta_{-i} \in \Theta_i}p_i(\vec\theta_{-i})\sigma_{-i}(\vec\theta_{-i}, a_{-i})\,d\vec\theta_{-i}
\end{align}

\begin{align*}
    \overline{U}_i(a_i, \vec\theta_i, \sigma_{-i}) &= \sum_{\vec\theta_{-i} \in \Theta_{-i}}p_i(\vec\theta_{-i})U_i(a_i, \sigma_{-i}(\vec\theta_{-i}), \vec\theta_i) \\
    &= \vec\theta_i \cdot (\overline{u}_{ij}(a_i, \sigma_{-i}))_{j \in J}
\end{align*}

When there are only two actions, the notion of threshold strategy remains exactly the same, but, since the type space is discrete, two strategies with different thresholds can evaluate to the same value for types $\vec\theta_i \in \Theta_i$ and actions $a_i \in A_i$.
\begin{dfn}[Equivalent strategies and thresholds]
    Two strategies $\sigma_i$ and $\sigma_i'$, for player $i\in I$, are said to be \textit{equivalent} ($\sigma_i \sim \sigma_i'$) if
    $$\forall (\vec\theta_i, a_i) \in \Theta_i \times A_i , \quad \sigma_i(\vec\theta_i, a_i) = \sigma_i'(\vec\theta_i, a_i)$$
    Given two threshold strategies $(\sigma_i, \vec\delta_i)$ and $(\sigma_i', \vec\delta_i')$, $\vec\delta_i$ and $\vec\delta_i'$ are said to be \textit{equivalent} ($\vec\delta_i \sim \vec\delta_i'$) if $\sigma_i \sim \sigma_i'$.
\end{dfn}
Notice that the binary relation $\sim$ restricted to threshold strategies is obviously an equivalence relation.

We also note that in contrast to the relation between strategies, the relation between thresholds is not an equivalence relation. Indeed, if a threshold leads to a mixed threshold strategy we cannot write $\vec\delta_i \sim \vec\delta_i$ because $\alpha_i$ of Definition~\ref{TS} is not constrained. In other words, two strategies with the same threshold $\vec\delta_i$ can be different (as long as the mixed case is reached by some $\theta_i \in \vec\Theta_i$). Note that a threshold strategy $\sigma_i$ is mixed if and only if there exists $\vec\theta_i \in \Theta_i$ such that $\vec\theta_i \cdot \vec\delta_i = 0$.\\

Given a threshold strategy $\sigma_i$, let $S_{eq}(\sigma_i): = \{\sigma_i' \mid \sigma_i' \sim \sigma_i\}$ be the set of threshold strategies equivalent to $\sigma_i$. Given a threshold $\vec\delta_i$, let $T_{eq}(\vec\delta_i) := \{ \vec\delta_i' \mid \vec\delta_i' \sim \vec\delta_i\}$ be  the set of thresholds equivalent to $\vec\delta_i$.

\begin{prop}
    Suppose a threshold $\vec\delta_i$, for $i\in I$, leads to a pure threshold strategy (i.e. $\forall \vec\theta_i \in \Theta_i$ $\vec\theta_i \cdot \vec\delta_i \neq 0$). Then, we have the following properties:
    \begin{enumerate}
        \item $T_{eq}(\vec\delta_i)$ is a non-empty set.
        \item For any (strictly) positive $\lambda \in \R^{*+}$, $\lambda \vec\delta_i \in T_{eq}(\vec\delta_i)$.
        \item $T_{eq}(\vec\delta_i)$ is a convex set.
        \item If $\Theta_i$ is finite, $T_{eq}(\vec\delta_i)$ contains vectors that are not collinear with $\vec\delta_i$.
    \end{enumerate}
\end{prop}

\begin{proof}
    \begin{enumerate}
        \item We have: $\vec\delta_i \in T_{eq}(\vec\delta_i)$ and thus the latter set is non-empty.
        \item Since the strategy only depends on the sign of $\vec\theta_i \cdot \vec\delta_i$, multiplying by a strictly positive $\lambda$ has no impact on the resulting action.
        \item Take $\vec \delta_i^1, \vec \delta_i^2 \in T_{eq}(\vec \delta_i)$, $\lambda \in [0, 1]$ and $\vec\delta_i^3 = \lambda \vec\delta_i^1 + (1 - \lambda)\vec\delta_i^2$, then notice that $\vec\delta_i^3 \cdot \vec\theta_i = \lambda (\vec\delta_i^1 \cdot \vec\theta_i) + (1 - \lambda)(\vec\delta_i^2 \cdot \vec\theta_i)$ has the same sign as $\vec\delta_i \cdot \vec\theta_i$, so that $\vec\delta_i^3 \in T_{eq}(\vec\delta_i)$. 
        \item Assume that $\Theta_i$ is finite and consider $\epsilon = \frac{1}{m+1}\displaystyle \min_{\vec\theta_i \in \Theta_i}\left| \vec\theta_i \cdot \vec\delta_i\right|$ and $\vec\epsilon = (\epsilon, \epsilon, ..., \epsilon)$. Let $\vec\delta_i' := \vec\delta_i + \vec\epsilon$:
        \begin{align*}
            \left| \vec\theta_i \cdot \vec\delta_i - \vec\theta_i \cdot \vec\delta_i' \right| = & \left| \vec\theta_i \cdot \vec\epsilon \right| \\
            \leq &\, m \epsilon \\
            < &  \min_{\vec\theta_i \in \Theta_i}\left| \vec\theta_i \cdot \vec\delta_i\right|
        \end{align*}
        So $\vec\theta_i \cdot \vec\delta_i$ and $\vec\theta_i \cdot \vec\delta_i'$ have the same sign for any $\vec\theta_i \in \Theta_i$ which concludes the proof.
    \end{enumerate}
\end{proof}

\subsection{The Double Game Prisoner Dilemma}

We keep the same framework for the DGPD with a discrete type space. For each $i\in I$, we have:

\begin{align}
    \overline{U}_i(a_i, \theta_i, \sigma_{-i}) &= \zeta_i^C(\sigma_{-i})U_i(a_i, C, \theta_i) + \zeta_i^D(\sigma_{-i})U_i(a_i, D, \theta_i)
\end{align}

and

\begin{align}
    \overline{U_i}(C, \theta_i, \sigma_{-i}) &= \zeta_i^C(\sigma_{-i}) \left[(1 - \theta_i) r + \theta_i y\right] + \zeta_i^D(\sigma_{-i}) \left[ (1 - \theta_i) s + \theta_i y\right]\\
    \overline{U_i}(D, \theta_i, \sigma_{-i}) &= \zeta_i^C(\sigma_{-i}) \left[(1 - \theta_i) t + \theta_i s\right] + \zeta_i^D(\sigma_{-i}) \left[ (1 - \theta_i) p + \theta_i s\right]
\end{align}

\begin{prop}
    Consider agent $i$ and let $\theta_i^x < \theta_i^{x+1}$ be two consecutive types from $\Theta_i$. All threshold strategies $\sigma_i$ with a threshold $\theta_i^*$ such that $\theta_i^x < \theta_i^* < \theta_i^{x+1}$ are equivalent.
\end{prop}
\begin{proof}
    this follows from the fact that there exists no type $\theta_i \in \Theta_i$ between $\theta_i^x$ and $\theta_i^{x+1}$, so as long as we keep the threshold between those two values we end up with the same actions for agent $i$.
\end{proof}

Note that for threshold strategies, the third case ($\theta_i = \theta_i^*$) can only be reached if $\theta_i^* \in \Theta_i$. So, if $\alpha_i \in \{0,1\}$ or $\theta_i^* \notin \Theta_i$ then the strategy is pure.\\

The search for pure Nash Equilibrium in the discrete case requires a different method since we do not have a general result for two-action multigames that can be used as in the case of continuous type space. To go further, we assume that the type space is finite. Recall, for player $i$, that $\zeta_i^{a_{-i}}(\sigma_{-i}) = \sum_{\theta_{-i} \in \Theta_{-i}}p_i(\theta_{-i})\sigma_{-i}(\theta_{-i}, a_{-i})$ is the probability that the other agent's action is $a_{-i}$ given that it follows $\sigma_{-i}$.

\begin{prop}[$\zeta_i^C$ monotonicity] \label{prop:zeta_monotonicity}
    Consider agent $i$ and suppose that the other agent follows a threshold strategy $\sigma_{-i}$ with threshold $\theta_{-i}^*$. This threshold strategy induces a value $\zeta_i^C(\sigma_{-i})$ for agent $i$. If the other agent changes its strategy $\sigma_{-i}^{\prime}$ by increasing its threshold $\theta^{* \prime}_{-i}$, then the induced value $\zeta_i^C(\sigma^{\prime}_{-i})$ decreases.
\end{prop}
\begin{proof}
    By increasing its threshold value, agent $-i$ decreases the probability of playing C and thus decreases the value $$\zeta_i^C(\sigma_{-i}') = \sum_{\theta_{-i} \in \Theta_{-i}} p_i(\theta_{-i})\sigma_{-i}'(\theta_{-i}, C)$$
\end{proof}

For integers $K,L$ with $K\leq L$, let $\llbracket K,L\rrbracket$ denote the set of integers, $K, K+1,\ldots,L$.
\begin{lemma} \label{fixedpointintervals}
Suppose $f:\llbracket0,M\rrbracket \rightarrow \llbracket0,N\rrbracket$ and $g:\llbracket0,N\rrbracket \rightarrow \llbracket0,M\rrbracket$ are both increasing (or both  decreasing). Then, there exists $c \in \llbracket0,N\rrbracket$ such that $f(g(c)) = c$.
\end{lemma}

\begin{proof}
We first notice that $h:=f\circ g$ is increasing and $0 \leq h(x) \leq N$ for $0\leq x\leq N$. Hence, there exists a least non-negative integer $k $ such that $c:=h^k(0)=h^{k+1}(0)$ and it follows that $h(c)=c$. 
\end{proof}

\begin{thm}\label{thm:finite_pure_ne}
For any DGPD with finite type space there exists a pure Bayesian Nash Equilibrium made of threshold strategies such that for both agents, the threshold $\theta_i^* \in [\min (\lambda, \mu), \max (\lambda, \mu)]$.
\end{thm}

\begin{proof}
    We suppose that $\mu > \lambda$ (the reasoning for $\mu < \lambda$ is similar). Suppose that agent $-i$ plays a threshold strategy $(\sigma_{-i}, \theta_{-i}^*)$. Then the agent $i$'s best response is also a threshold strategy $(\sigma_i, \theta_i^*)$. By the monotonicity of $\theta_i^*$ and $\zeta_i^C$ (Propositions~\ref{prop:threshold_monotonicity} and \ref{prop:zeta_monotonicity}), if agent $-i$ increases its threshold then the threshold of agent $i$'s best response also increases (if $\mu > \lambda$ then agent $i$'s best response threshold decreases).

    Now consider the partition of $[0,1]$ into intervals according to agents' type spaces:
    \begin{align}
        &\Theta_i = \{\theta_i^1, \theta_i^2, ..., \theta_i^{n_i}\} \quad \text{with} \quad 0 \leq \theta_i^1 < \theta_i^2 < ... < \theta_i^{n_i} \leq 1\\
        &I_i^0 = [0, \theta_i^1], \quad I_i^1 = [\theta_i^1, \theta_i^2], \quad ... \quad I_i^{n_i} = [\theta_i^{n_i}, 1]
    \end{align}
    
    For agent $i$'s threshold strategy $(\sigma_i, \theta_i^*)$ and agent $-i$'s best response $(\sigma_{-i}, \theta_{-i}^*)$, there exists (1) $k_i \in \llbracket 0, n_i \rrbracket$ such that $\theta_i^* \in I_i^{k_i}$ and (2) $k_{-i} \in \llbracket 0, n_{-i} \rrbracket$ such that $\theta_{-i}^* \in I_{-i}^{k_{-i}}$. Note that if $\theta_{-i}^* \in \Theta_{-i}$, then it belongs to two adjacent intervals $I_{-i}^{k_{-i}}$ and $I_{-i}^{k_{-i}+1}$. In this case we arbitrarily choose to take $I_{-i}^{k_{-i}}$.
    
    We define the transition functions $t_{1 \rightarrow 2}: \llbracket 0, n_1 \rrbracket \rightarrow \llbracket 0, n_2 \rrbracket$ and $t_{2 \rightarrow 1}: \llbracket 0, n_2 \rrbracket \rightarrow \llbracket 0, n_1 \rrbracket$ such that $k_{-i} = t_{i \rightarrow -i}(k_i)$. In order to show that there exists a pure Nash Equilibrium, we need to show that:
    \[ \exists (k_1,k_2)\in \llbracket 0,n_1 \rrbracket \times \llbracket 0,n_2\rrbracket \mid k_1 = t_{2 \rightarrow 1}(k_2), \quad k_2 = t_{1 \rightarrow 2}(k_1) \]
    
    Or equivalently we search for $k_i$ such that $k_i = t_{-i \rightarrow i}(t_{i \rightarrow -i}(k_i))$. By Lemma~\ref{fixedpointintervals} with $t_{1 \rightarrow 2}$ and $t_{2 \rightarrow 1}$ as $f$ and $g$, we can find a solution $k_i$.
\end{proof}

\subsection{Example}

Suppose that $G$ is a Double Game Prisoner Dilemma with uniform prior such that  utilities for both agents are given according to Table~\ref{example:table1} and \ref{example:table2} and type space	$\Theta_i = \{t/60: t \in \llbracket 0, 60 \rrbracket \}$ for $i=1,2$. We obtain $\mu=1/5$ and $\lambda=1/4$.

\begin{table}[!htb]
   \begin{minipage}{0.4\textwidth}
     \centering
     \begin{tabular}{|c|c|c|}
        \cline{2-3}
        \multicolumn{1}{c|}{} & $C$ & $D$\\\hline
        C &$(16, 16) $&$(3, 20)$\\\hline
        D &$(20, 3) $&$(6, 6)$\\\hline
    \end{tabular}
    \caption{Utilities for the 1\textsuperscript{st} game}
    \label{example:table1}
    
    \vspace{10mm}
    
    \begin{tabular}{|c|c|c|}
        \cline{2-3}
        \multicolumn{1}{c|}{} & $C$ & $D$\\\hline
        C& $(15, 15) $&$(15,3 )$ \\\hline
        D &$(3, 15) $&$(3, 3)$\\\hline
    \end{tabular}    
    \caption{Utilities for the 2\textsuperscript{nd} game}
    \label{example:table2}
   \end{minipage}\hfill
   \begin{minipage}{0.6\textwidth}
     \centering
     \begin{tikzpicture}[scale=.6]
        \fill[blue!10!white] (0,0) rectangle (3,3);
        \fill[blue!10!white] (6,0) rectangle (9,3);
        \fill[blue!10!white] (0,6) rectangle (3,9);
        \fill[blue!10!white] (6,6) rectangle (9,9);
        \draw[black, thick] (0,0) -- (0,9);
        \draw[black, thick] (3,0) -- (3,9);
        \draw[black, thick] (6,0) -- (6,9);
        \draw[black, thick] (9,0) -- (9,9);
        \draw[black, thick] (0,0) -- (9,0);
        \draw[black, thick] (0,3) -- (9,3);
        \draw[black, thick] (0,6) -- (9,6);
        \draw[black, thick] (0,9) -- (9,9);
        \draw[dashed] (4,0) -- (4,9);
        \draw[dashed] (5,0) -- (5,9);
        \draw[dashed] (0,4) -- (9,4);
        \draw[dashed] (0,5) -- (9,5);
        
        \draw(0,3) node[left] {$\tfrac{1}{5}$};
        \draw(0,6) node[left] {$\tfrac{1}{4}$};
        \draw(0,9) node[left] {$1$};
        \draw(0,0) node[below] {$0$};
        \draw(3,0) node[below] {$\tfrac{1}{5}$};
        \draw(6,0) node[below] {$\tfrac{1}{4}$};
        \draw(9,0) node[below] {$1$};
        \draw(4,0) node[below] {$\tfrac{13}{60}$};
        \draw(5,0) node[below] {$\tfrac{14}{60}$};
        \draw(0,4) node[left] {$\tfrac{13}{60}$};
        \draw(0,5) node[left] {$\tfrac{14}{60}$};

        \node at (1.5,7.5) {\footnotesize$(D,C)$};	
        \node at (1.5,1.5) {\footnotesize$(D,D)$};		
        \node at (7.5,7.5) {\footnotesize$(C,C)$};		
        \node at (7.5,1.5) {\footnotesize$(C,D)$};
    \end{tikzpicture}
     \caption{Representation of both agent's type space}
   \end{minipage}
\end{table}

To find a pure Bayesian NE, one has to find a pair $(s_1, s_2)$ of threshold strategies such that $s_1$ is a best response to $s_2$ and $s_2$ is a best response to $s_1$. By symmetry of the game, we can expect that $s_1 = s_2$ (i.e. they have the same threshold). We have:

\begin{align*}
    \theta_i^*(\zeta_i^C) = \frac{\zeta_i^C(20 - 16) + (1 - \zeta_i^C)(6 - 3)}{\zeta_i^C(20 - 16) + (1 - \zeta_i^C)(6 - 3) + (15 - 3)} = \frac{\zeta_i^C + 3}{\zeta_i^C + 15}
\end{align*}

while $\zeta_i^C$ can take 3 different values: 

\begin{align*}
    \theta_{-i}^* \in \left]\frac{1}{5}, \frac{13}{60}\right[ \implies \zeta_i^C = \frac{48}{61} \\
    \theta_{-i}^* \in \left]\frac{13}{60}, \frac{14}{60}\right[ \implies \zeta_i^C = \frac{47}{61} \\
    \theta_{-i}^* \in \left]\frac{14}{60}, \frac{1}{4}\right[ \implies \zeta_i^C = \frac{46}{61}
\end{align*}

Possible values for $\theta_i^*$ are $\theta_i^*(\frac{48}{61}) \approx 0.239875$, $\theta_i^*(\frac{47}{61}) \approx 0.239085$ or $\theta_i^*(\frac{46}{61}) \approx 0.238293$. Those three values are in the interval $]\frac{14}{60}, \frac{1}{4}[$. So the pair $(s_1, s_2)$ where

\begin{align*}
s_1(\theta_1^*) = s_2(\theta_2^*) = 
    \begin{cases}
    \text{$D$ if $\theta_i^* < \frac{1}{4}$}\\
    \text{$C$ if $\theta_i^* \geq \frac{1}{4}$}
    \end{cases}
\end{align*}

is a pure (Bayesian) strategy Nash Equilibrium for the game $G$.

\subsection{On a simple multigame classification}

In this section, we consider $G_{2,2,2}$ multigames with finite type space. Both actions are still denoted by C and D (even if they are not necessarily associated to cooperation or defection). The payoff matrix $U$ is not constrained as opposed to DGPD configuration. First, we define a simple classification of such games according to properties of $U$.

\begin{dfn}
    We denote the \textit{type space configuration} by $(\Theta_i, p_i)$. Any payoff matrix $U$ belongs to one of the following sets:
    \begin{enumerate}
        \item The \textit{full set}: payoff matrices $U$ such that for any type space configuration the game $G=\langle U,(\Theta_i, p_i)\rangle$ has a pure Nash Equilibrium.
        \item The \textit{solutionless set}: payoff matrices $U$ such that for any type space configuration the game $G=\langle U,(\Theta_i, p_i)\rangle $ has no pure Nash Equilibrium.
        \item The \textit{hybrid set}: payoff matrices $U$ such that the existence of a pure NE for $G=\langle U,(\Theta_i, p_i)\rangle $ depends on the type space configuration.
    \end{enumerate}
\end{dfn}

\begin{prop}
    The full set is not empty: it contains matrices $U$ with the DGPD payoff constraints.
    The hybrid set is also not empty.
\end{prop}
\begin{proof}
    The first set is obviously not empty in view of Theorem~\ref{thm:finite_pure_ne}. The second set is not empty as we can find $G$ and $G'$ sharing the same payoff matrix such that one has a pure NE and the other not (see next examples).
\end{proof}

Consider the uniform double game $G$ with utilities for agents $i= 1,2$ given by Table \ref{tab:double_game_payoff} with type space $\Theta_1 = \{0.2, 0.5, 0.6\}$, $\Theta_2 = \{0.2, 0.4, 0.8\}$ and the prior $p_1$, $p_2$ such that:

\begin{align*}
    p_1(\theta_1 = 0.2) = 0.3 \quad p_1(\theta_1 = 0.5) = 0.4 \quad p_1(\theta_1 = 0.6) = 0.3\\
    p_2(\theta_2 = 0.2) = 0.3 \quad p_2(\theta_2 = 0.4) = 0.4 \quad p_2(\theta_2 = 0.8) = 0.3
\end{align*}

\begin{table}[h]
		\centering
		\begin{minipage}{.2\linewidth}
			\begin{tabular}{|c|c|c|}
				\cline{2-3} 
				\multicolumn{1}{c|}{} & $C$ & $D$  \\ \hline
				C &$(3, 4) $&$(4, 3)$\\\hline
				D &$(5, 0) $&$(1, 1)$\\\hline
			\end{tabular}
		\end{minipage}
		\hspace{20mm}
		\begin{minipage}{.2\linewidth}
			\begin{tabular}{|c|c|c|}
				\cline{2-3} 
				\multicolumn{1}{c|}{} & $C$ & $D$  \\ \hline
				C& $(2, 4) $&$(6, 5)$ \\\hline
				D &$(7, 2) $&$(5, 2)$\\\hline
			\end{tabular}
		\end{minipage}	
		\caption{Utilities for both basic games.}
		\label{tab:double_game_payoff}
	\end{table}	
	
This game has no pure Bayesian Nash Equilibrium. Consider a slight variation $G'$ of the game $G$ such that utilities for agents $i=1,2$ are given by Table \ref{tab:double_game_payoff} with type space $\Theta_1 = \{0.2, 0.5, 0.7\}$, $\Theta_2 = \{0.2, 0.4, 0.8\}$ and the prior $p_1$, $p_2$ such that:

\begin{align*}
    p_1(\theta_1 = 0.2) = 0.3 \quad p_1(\theta_1 = 0.5) = 0.4 \quad p_1(\theta_1 = 0.7) = 0.3\\
    p_2(\theta_2 = 0.2) = 0.3 \quad p_2(\theta_2 = 0.4) = 0.4 \quad p_2(\theta_2 = 0.8) = 0.3
\end{align*}

Consider $\sigma_1$ a C/D threshold strategy with $\theta_1^* \in ]0.5, 0.7[$ and $\sigma_2$ a C/D threshold strategy with $\theta_2^* \in ]0.2, 0.4[$. The pair $(\sigma_1,\sigma_2)$ is a pure Bayesian Nash Equilibrium for the game $G'$. \\

Note that we don't give any particular property for the solutionless set. We postulate that this set is empty (see Section~\ref{ALRe} on algorithmic results). Recall the following formulae for the $G_{2,2,2}$ configuration:

\begin{align}
    \overline{U_i}(a_i, \theta_i, \sigma_{-i}) &= (1 - \theta_i) \overline{u}_{i1}(a_i, \sigma_{-i}) + \theta_i \overline{u}_{i2}(a_i, \sigma_{-i})\\
    \overline{u}_{ij}(a_i, \sigma_{-i}) &= \zeta_i^C(\sigma_{-i}) u_{ij}(a_i, C) + (1 - \zeta_i^C(\sigma_{-i})) u_{ij}(a_i, D)
\end{align}

Observe that with no assumption on the payoff matrix, there's no guarantee that $\overline{U_i}(C, \theta_i, \sigma_{-i})$ and $\overline{U_i}(D, \theta_i, \sigma_{-i})$ will cross for a given $\sigma_{-i}$. Moreover, the crossing point $\theta_i^*$ (if it exists) is not guaranteed to be inside $[0, 1]$ as illustrated by Figure~\ref{example:threshold_inside} and Figure~\ref{example:threshold_outside}. It cuts $\R$ into two regions, one in which the best action is C and the other in which the best action is D. If the left region is C then the resulting strategy is a \textit{C/D strategy} (Figure~\ref{example:threshold_inside}) and if the left region is D then the resulting strategy is a \textit{D/C strategy} (Figure~\ref{example:threshold_outside}). We call this the {\em strategy type}. For any DGPD configuration, a best response is always a D/C strategy.

\begin{figure}[!htb]
   \begin{minipage}{0.5\textwidth}
     \centering  
     \includegraphics[width = \hsize]{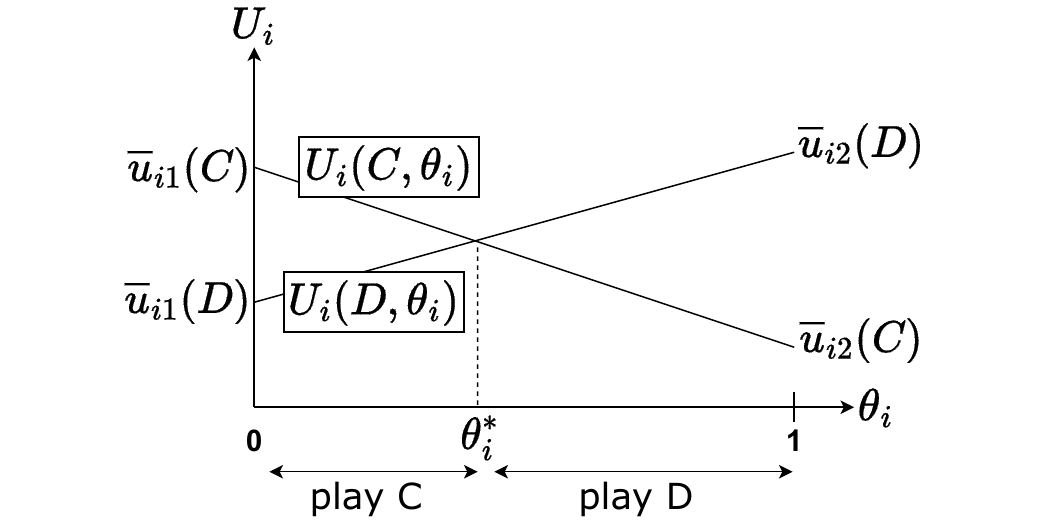}
    \caption{Crossing point (threshold) inside $[0, 1]$}
    \label{example:threshold_inside}
   \end{minipage}
   \begin{minipage}{0.5\textwidth}
     \centering
   \includegraphics[width = \hsize]{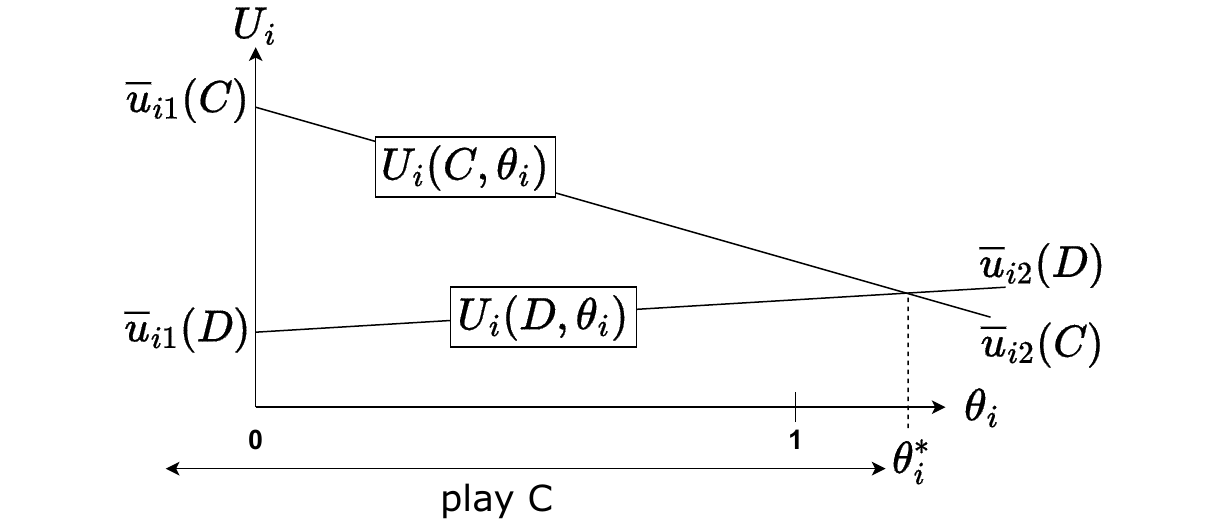}
     \caption{Crossing point outside $[0, 1]$}
     \label{example:threshold_outside}
   \end{minipage}
\end{figure}

As we did for the DGPD, we define the threshold function $\theta_i^*(\sigma_{-i})$ for agent $i$:

\begin{align}
    \theta_i^*(\sigma_{-i}) = \frac{\overline{u}_{i1}(C, \sigma_{-i}) - \overline{u}_{i1}(D, \sigma_{-i})}{(\overline{u}_{i1}(C, \sigma_{-i}) - \overline{u}_{i1}(D, \sigma_{-i})) - (\overline{u}_{i2}(C, \sigma_{-i}) - \overline{u}_{i2}(D, \sigma_{-i}))}
\end{align}

When $\overline{u}_{i1}(C, \sigma_{-i}) - \overline{u}_{i1}(D, \sigma_{-i}) = \overline{u}_{i2}(C, \sigma_{-i}) - \overline{u}_{i2}(D, \sigma_{-i})$, both utility functions have the same slopes and the threshold function is not defined. As long as $\overline{u}_{i1}(C, \sigma_{-i}) \neq \overline{u}_{i1}(D, \sigma_{-i})$, utility functions are not equal so there is one best action (either C or D). In this case the best response is also a threshold strategy with the threshold $\theta_i^* = \pm \infty$. Otherwise, both utility functions are equal (there's an infinite number of crossing points). Because of the latter case, we cannot state that any best response is made of threshold strategies anymore.

\begin{dfn}[D/C threshold strategy]
    A threshold D/C strategy $\sigma_i$ with threshold $\theta_i^* \in \R\cup\{-\infty, +\infty\}$ for agent $i$ is a strategy such that:
    \begin{align}
        \sigma_i(\theta_i) = 
        \begin{cases}
            D & \text{$\theta_i < \theta_i^*$}\\
            C & \text{$\theta_i > \theta_i^*$}\\
            \alpha_i C + (1 - \alpha_i) D & \text{$\theta_i = \theta_i^*$}
        \end{cases}
    \end{align}
\end{dfn}

Now we can express the threshold as a function of $\zeta_i^C \in [0,1]$ instead of $\sigma_{-i}$:

\begin{align}
    \theta_i^*(\zeta_i^C) &= \frac{\delta_{i1}^D + \zeta_i^C(\delta_{i1}^C - \delta_{i1}^D)}{\delta_{i1}^D - \delta_{i2}^D + \zeta_i^C(\delta_{i1}^C - \delta_{i1}^D + \delta_{i2}^D - \delta_{i2}^C)}\\
    \mbox{where }\delta_{ij}^{a_{-i}} &:= u_{ij}(C, a_{-i}) - u_{ij}(D, a_{-i}) 
\end{align}

The case of equal slopes is reached at the forbidden value $\widetilde \zeta_i^C$:

\begin{align}
    \widetilde \zeta_i^C := -\frac{\delta_{i1}^D - \delta_{i2}^D}{\delta_{i1}^C - \delta_{i2}^C + \delta_{i2}^D - \delta_{i1}^D} 
\end{align}

And therefore the case of equal utility functions happens when: $$\widetilde \zeta_i^C = \frac{\delta_{i1}^D}{\delta_{i1}^D - \delta_{i1}^C}$$

We notice that the graph of $\theta_i^*(\zeta_i^C)$ is split into two regions by the forbidden value (Figure~\ref{example:inside} and Figure~\ref{example:outside}). Each region is associated with a different strategy type. Thus, if $\widetilde \zeta_i^C$ is not in $[0, 1]$ then agent $i$ will always play the same strategy type notwithstanding its opponent's strategy. In Figure~\ref{example:outside} the forbidden value is outside $[0,1]$ so agent $i$ will only play D/C strategy. In Figure~\ref{example:inside}, the forbidden value is in $[0,1]$, so agent $i$ may play both strategy types depending on the opponent's strategy $\sigma_{-i}$: if $\zeta_i^C(\sigma_{-i}) < \widetilde \zeta_i^C$ then agent $i$ plays a C/D strategy and if $\zeta_i^C(\sigma_{-i}) > \widetilde \zeta_i^C$ then agent $i$ plays a D/C strategy. Suppose that $\widetilde \zeta_i^C \notin [0,1]$ for $i=1,2$. Both agents always stick to the same strategy type and there are 3 cases: (1) both play C/D, (2) both play D/C (like the DGPD) and (3) one plays C/D while the other plays D/C.

\begin{figure}[!htb]
   \begin{minipage}{0.5\textwidth}
     \centering  
     \includegraphics[width = \hsize]{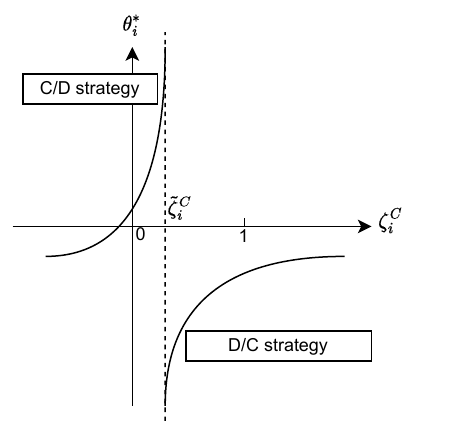}
    \caption{Forbidden value inside $[0, 1]$}
    \label{example:inside}
   \end{minipage}
   \begin{minipage}{0.5\textwidth}
     \centering
   \includegraphics[width = \hsize]{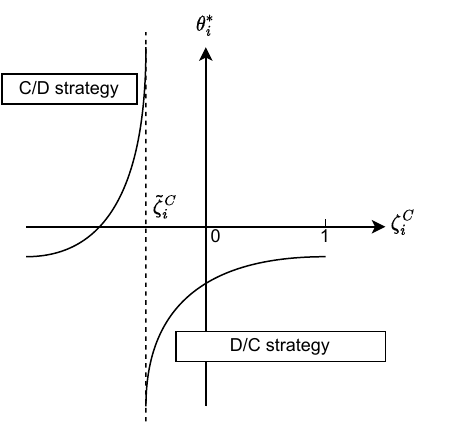}
     \caption{Forbidden value outside $[0, 1]$}
     \label{example:outside}
   \end{minipage}
\end{figure}

In order to characterize the variations of $\theta_i^*$, define:

\begin{align}
    \Delta_i := \delta_{i1}^D \delta_{i2}^C - \delta_{i1}^C \delta_{i2}^D
\end{align}

\begin{prop}[$\theta_i^*$ monotonicity]
    The threshold function $\theta_i^*(\zeta_i^C)$ is monotonic and satisfies:
    \begin{enumerate}
        \item if $\Delta_i > 0$ then $\theta_i^*(\zeta_i^C)$ is increasing,
        \item if $\Delta_i < 0$ then $\theta_i^*(\zeta_i^C)$ is decreasing,
        \item if $\Delta_i = 0$ then $\theta_i^*(\zeta_i^C)$ is constant.
    \end{enumerate}
\end{prop}
\begin{proof}
Simply compute the derivative of the function $\theta_i^*(\zeta_i^C)$ which is an homography  $f: x \rightarrow \frac{ax+b}{cx+d}$ and get the condition on $\Delta_i$.
\end{proof}

\begin{prop}\label{prop:full_matrix_conditions}
    Let $U$ be the payoff matrix of a $2$-player double game such that:
    \begin{enumerate}
        \item $\widetilde \zeta_i^C \notin [0,1]$ for $i=1,2$ where $\widetilde \zeta_i^C$ is the forbidden value of $\theta_i^*(\zeta_i^C)$,
        \item $\Delta_1$ and $\Delta_2$ have the same sign if both agents play the same strategy type, or have opposite sign otherwise.
    \end{enumerate}
    Then for any type space configuration $(\Theta_i, p_i)$ the game $G=\langle U, (\Theta_i, p_i)\rangle $ has a pure Bayesian Nash Equilibrium.
\end{prop}
\begin{proof}
    The two conditions imply that (1) agents have a unique strategy type and (2) their best response threshold functions have the same monotonicity, i.e., increasing or decreasing. Thus we can follow the same reasoning as we had for DGPD to prove the existence of a pure Bayesian NE.
\end{proof}

\subsection{Algorithmic results}\label{ALRe}

In this section, we develop efficient algorithms to find pure Bayesian Nash Equilibrium for finite type space $G_{2,2,2}$. The first part focuses on an optimized version for DGPD while the second one focuses on a more general version. In the third part, the complexity of both algorithms are evaluated and compared to each other. 

\subsubsection{Algorithm for DGPD}

Recall that the agents' best response sets are only made of threshold strategies. Thus, given a finite type space (with $n_i$ elements for agent $i$) the search space is made of $n_1 \times n_2$ threshold strategy pairs. The pure Bayesian Nash Equilibrium search consists in finding a fixed point (two strategies that are mutually best responses of each other) among those combinations.

\begin{figure}[ht]
\centering
\includegraphics[width = 0.7\hsize]{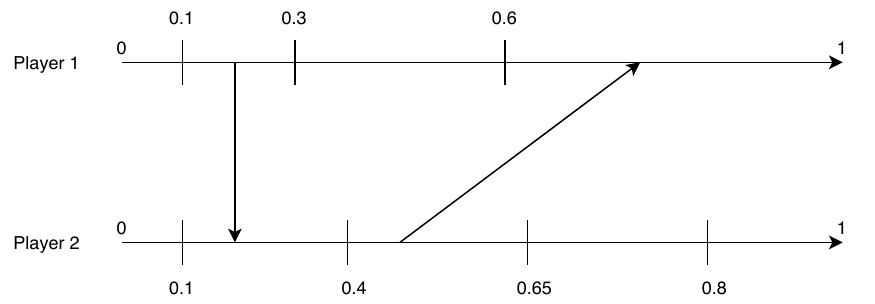}
\caption{If agent 1 plays a threshold strategy with $\theta_1^* \in [0.1, 0.3]$, agent 2's best response is a threshold strategy with $\theta_2^* \in [0.1, 0.4]$}
\label{fig:search_5}
\end{figure}

For clarity, we formulate a graphical method to represent the solution search that we call strategy diagram and looks like figure \ref{fig:search_5}. There are two unit intervals $[0,1]$, one for each agent's type space. An arrow from agent $i$'s interval $I_i$ to agent $-i$'s interval $I_{-i}$ indicates that if agent $i$ plays a threshold strategy with a threshold in $I_i$, the best response of agent $-i$ is a threshold strategy with a threshold in $I_{-i}$. A solution is then simply represented by two \textit{compatible arrows} as displayed in Figure~\ref{fig:search_6}.

\begin{figure}[ht]
\centering
\includegraphics[width = 0.7\hsize]{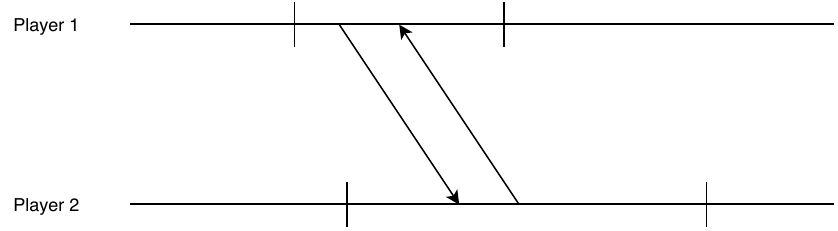}
\caption{A pure Bayesian NE represented by compatible arrows.}
\label{fig:search_6}
\end{figure}

We now introduce algorithm \ref{alg:non_symmetric} for NE search on finite DGPD. For every threshold strategy of agent $1$, we compute the associated best response of agent $2$ and then compute the best response of agent $1$ given the latter. Whenever for any of agent $1$'s threshold strategy the computed best response is equivalent, we obtain a pure NE. The procedure \texttt{compute\_cumul\_proba()} returns the cumulative probabilities given a probability distribution. \texttt{finder()} returns the index of the type space interval that contains a given threshold. \texttt{search\_space\_boundaries()} computes $\alpha = min(\mu, \lambda)$ and $\beta = max(\mu, \lambda)$ and then the associated indices in the type space. It helps us optimize the overall algorithm by reducing the search space to thresholds $\theta_i^* \in [\alpha, \beta]$. Finally, \texttt{threshold\_i()} is the threshold function of agent $i$.

\begin{algorithm}
\caption{Exhaustive NE search}
\begin{algorithmic}
\REQUIRE $t, r, y, p, s, \Theta_1, \Theta_2, p_1, p_2$

\STATE $cumul\_proba\_1 \leftarrow compute\_cumul\_proba(p_1)$
\STATE $cumul\_proba\_2 \leftarrow compute\_cumul\_proba(p_2)$
\STATE $start\_1, end\_1 \leftarrow search\_space\_boundaries(t, r, y, p, s, \Theta_1)$
\STATE
\FOR{$start\_1 \leq i < end\_1 + 1$}
\STATE $\zeta^C_2 \leftarrow 1 - cumul\_proba\_1[i]$
\STATE $\theta^*_2 \leftarrow threshold\_2(\zeta^C_2)$
\STATE $k \leftarrow finder(\Theta_2, \theta^*_2)$
\STATE $\zeta^C_1 \leftarrow 1 - cumul\_proba\_2[k]$
\STATE $\theta^*_1 \leftarrow threshold\_1(\zeta^C_1)$

\IF{$\Theta_1[i-1] \leq \theta^*_1 \leq \Theta_1[i]$}
\RETURN $\theta^*_1$, $\theta^*_2$
\ENDIF

\ENDFOR
\RETURN False

\end{algorithmic}
\label{alg:non_symmetric}
\end{algorithm}

Figures \ref{fig:search_1} and \ref{fig:search_2} illustrate the finding of NE search for two situations that we encounter. Note that when $\lambda < \mu$ we reverse the $2^{nd}$ agent's axis as the best response threshold monotonicity is reversed compared to $\lambda > \mu$. Also note that there could be multiple solutions (Figure \ref{fig:search_2}). If we only seek one solution then we stop at the first finding to reduce the computational cost.

\begin{figure}[ht]
\centering
\includegraphics[width = 0.7\hsize]{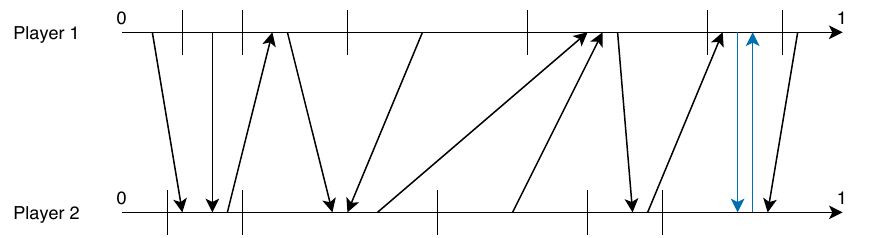}
\caption{Strategy diagram when $\lambda > \mu$}
\label{fig:search_1}
\end{figure}

\begin{figure}[ht]
\centering
\includegraphics[width = 0.7\hsize]{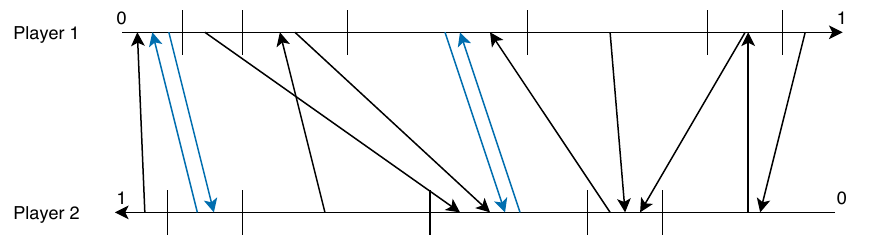}
\caption{Strategy diagram when $\lambda < \mu$}
\label{fig:search_2}
\end{figure}

\subsubsection{General algorithm}

We now describe a different algorithm that can handle a broader range of games but is more expensive. It can also handle DGPD games but is less optimal.

\begin{figure}[h]
\centering
\includegraphics[width = 0.95\hsize]{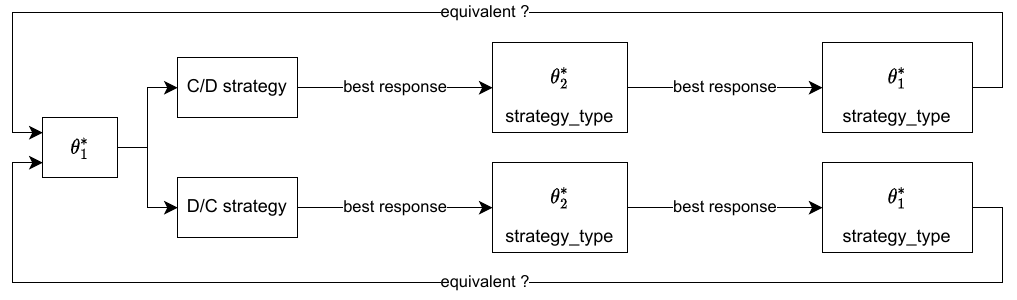}
\caption{Search process for the discrete Double Game}
\label{fig:double_game_process}
\end{figure}

Apart from the fact that we have to explore both C/D and D/C strategies, the double game algorithm is very similar to the DGPD algorithm. For each type space interval of agent 1, we compute the best response of agent 2 and the best response of agent 1 given the latter. We check that we fall back to the initial strategy by ensuring that the thresholds are in the same interval and that the strategy types are the same. Of course, we cannot benefit from the reduced search space with $\lambda$ and $\mu$ as they make no sense in the general context. The overall procedure is summarized in Figure~\ref{fig:double_game_process}.

Given this general algorithm we explore some properties of payoff matrices. We generated hundreds of random payoff matrices and thousands of type space configurations and ran Nash Equilibrium searches. We experimentally classified the payoff matrices thanks to the NE search results. Interestingly, we notice that none of the generated matrices belongs to the solutionless set as we postulated earlier. Secondly, we find that there exist matrices in the full set that do not satisfy the conditions of Proposition~\ref{prop:full_matrix_conditions}: the set of conditions is sufficient, but not necessary, for belonging to the full set.

\begin{figure}[h]
    \centering
    \includegraphics[width = 0.49\hsize]{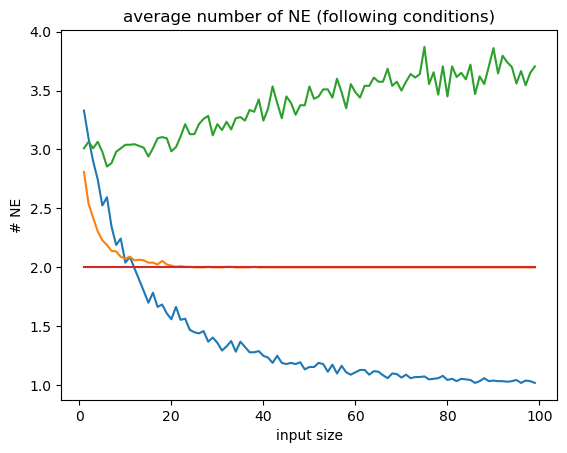}
    \includegraphics[width = 0.49\hsize]{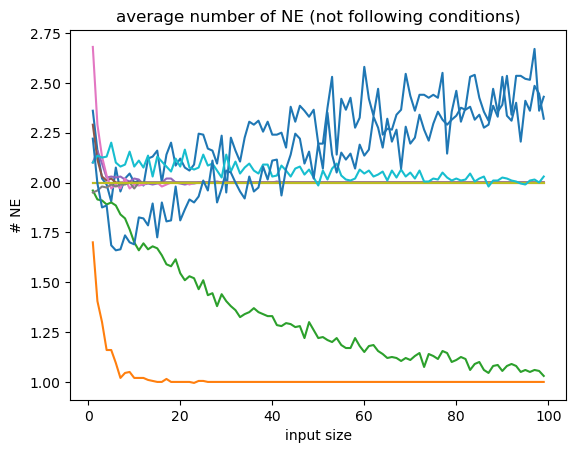}
    \caption{Average number of solutions with respect to the input size of both agents. Each curve corresponds to a payoff matrix. Among matrices belonging to the full set, we distinguish those satisfying the conditions of Proposition~\ref{prop:full_matrix_conditions} (left) from those not satisfying them (right).}
    \label{fig:full_matrix_avg_solutions}
\end{figure}

Next, using a modified version of the NE search that doesn't stop at the first solution (if it exists) we computed the average number of solutions for different type space configurations. Let {\em input size} denote  the number of elements in the type space for both players. In this context, $n = card(\Theta_1)$ and $m = card(\Theta_2)$.

Figure~\ref{fig:full_matrix_avg_solutions} shows the results for a sample of payoff matrices randomly picked from the full set (thanks to our previous classification). The first trivial result is that the average never goes under 1 (otherwise it wouldn't belong to the full set). Secondly, it seems that the average number of solution can either be constant, increase or decrease with the input size. In many cases we even observe that the average number tends to stabilize around an arbitrary value which seems to be an integer.

\begin{figure}[h]
    \centering
    \includegraphics[width = 0.55\hsize]{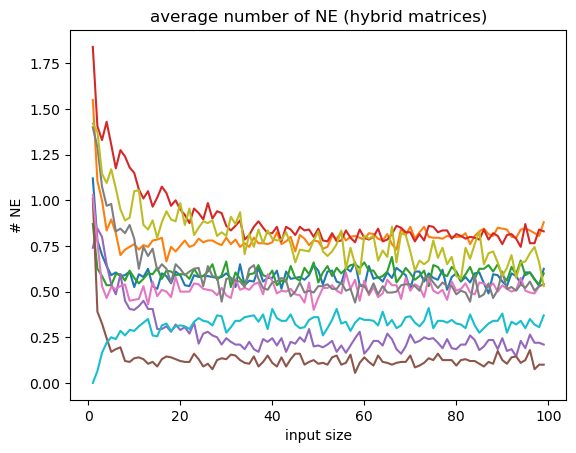}
    \caption{Average number of solutions with respect to the input size of both agents. Each curve corresponds to a payoff matrix from the hybrid set.}
    \label{fig:hybrid_matrix_avg_solution}
\end{figure}

We repeated this experiment with payoff matrices randomly picked from the hybrid set. Figure~\ref{fig:hybrid_matrix_avg_solution} shows a sample of results for such matrices. As we said earlier for the full set, it seems that the averages can either increase, decrease or stay constant and that they converge to different values. This time, those values are smaller than 1 and may not be integers.

\subsubsection{Complexity comparison}

In this part, we discuss the complexity of each algorithm and compare their performances.

First of all, the left part of Figure~\ref{fig:dgpd_complexity_comparison} displays the complexity of DGPD NE search for different variations of inputs. When $n=m$, both input sizes vary, while when $n=10$ or $m=10$ only input size varies. Notice that when the first agent's input size is constant, the complexity is also constant. When $n$ varies, the complexity is linear and is almost independent of whether $m$ varies or not. In fact, the main driver of the complexity is the main loop that iterates over the elements of $\Theta_1$. The procedure \texttt{finder()} has a low computational cost as it consists in a binary search. In practice, for imbalanced type space size, one should always consider the agent having the smallest type space as the first agent.

\begin{figure}[h]
    \centering
    \includegraphics[align=t,width = 0.49\hsize]{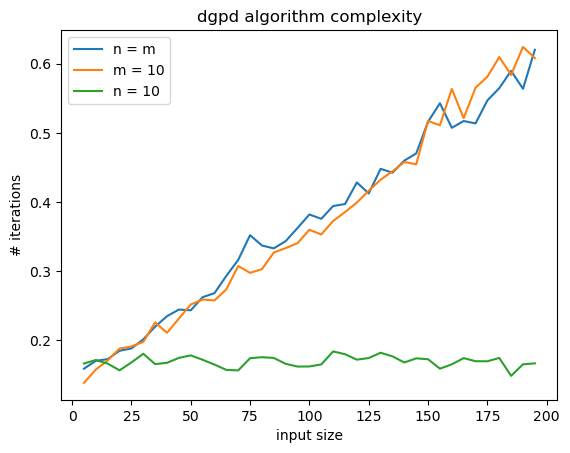}
    \includegraphics[align=t,width = 0.49\hsize]{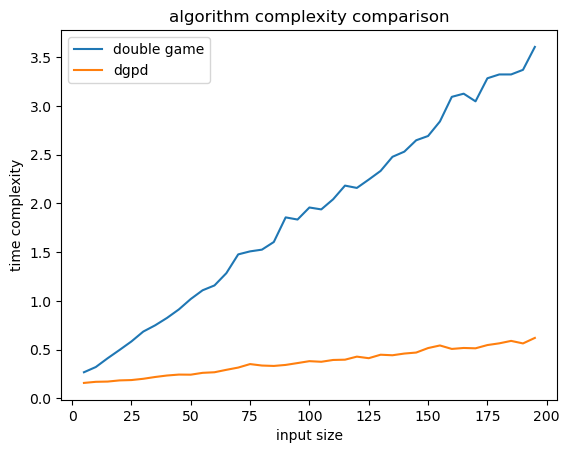}
    \caption{DGPD time complexity (left) and comparison with general algorithm (right)}
    \label{fig:dgpd_complexity_comparison}
\end{figure}

The right part of Figure~\ref{fig:dgpd_complexity_comparison} illustrates the improved performance of DGPD specific algorithm over the more general one. The latter also has a linear complexity but with a much higher slope. One explanation comes from the fact that for each element of $\Theta_1$ we consider both C/D and D/C strategies as starting points. Also, we don't reduce the search space with bounds like $\mu$ and $\lambda$.

\begin{figure}[h]
    \centering
    \includegraphics[align=t,width = 0.49\hsize]{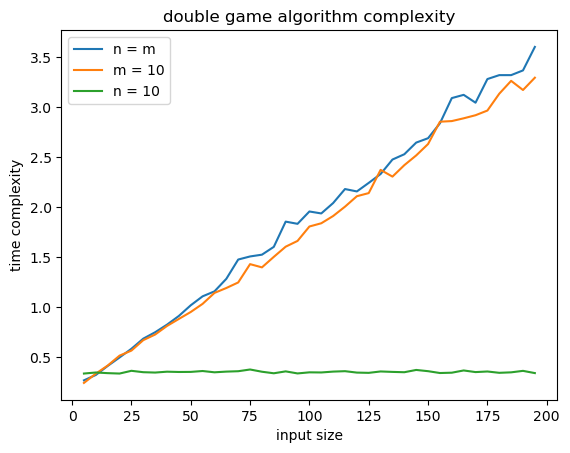}
    \includegraphics[align=t,width = 0.49\hsize]{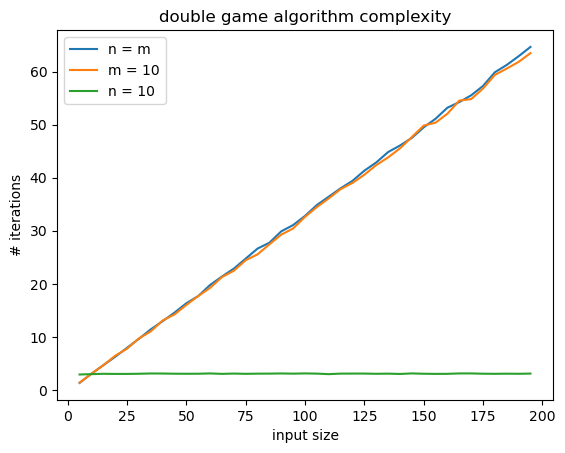}
    \caption{Complexity comparison for different input sizes through computation time (left) and average number of iterations (right)}
    \label{fig:double_game_complexity_comparison}
\end{figure}

For the general algorithm, we also compare the complexity for different variations of the input size as display in Figure~\ref{fig:double_game_complexity_comparison} (left). Again the complexity heavily relies on the first agent's type space size. In contrast to the DGPD algorithm, we notice that for $m=10$ the complexity is slightly below the complexity when $n=m$. On the right side of Figure~\ref{fig:double_game_complexity_comparison} we see that this behaviour cannot be explained by a difference in the number of iterations on the main loop. It is probably due to the cost of \texttt{finder()} that increases in $\mathcal{O}(\log(m))$ when $m$ increases. This effect might also affect the DGPD algorithm but is not significant in our experiment.

\section{Conclusion}

In this paper, we explored pure Bayesian Nash Equilibrium existence for a subset of uniform multigames. We made the distinction between games according to their type space, either continuous or discrete.

For continuous type space games with $2$ actions, we showed in Theorem \ref{thm:practical_continuous_theorem} the existence of a pure Bayesian Nash Equilibrium when there are local games having a strictly dominant strategy for each agent. We illustrated its application through the DGPD and SDAP examples that can model real situations with more precision than toy examples usually presented. In Section \ref{algRes_cont} we formulated a methodology to solve $2$ action games with any kind of prior.

For finite type space DGPD, we showed in Theorem \ref{thm:finite_pure_ne} the existence of a pure Bayesian Nash Equilibrium. Following this, we were able to provide efficient algorithms to find pure Bayesian Nash Equilibrium and explore experimentally our classification of discrete Double Multigames (Proposition \ref{prop:full_matrix_conditions}).

Threshold Strategy is a core concept developed for both the continuous and the discrete type space games with $2$ actions. As we saw, a threshold strategy is fully characterized by its threshold and defines 3 regions in the type space: one associated to a pure action C, one associated to a pure action D and one associated to a mix. By construction, a best response must be a threshold strategy. The threshold strategies presented for the DGPD are the more basic version where we have: play D if $\theta_i < \theta_i^*$, play C if $\theta_i > \theta_i^*$ and a mix of both if $\theta_i = \theta_i^*$.

As for future work, we could try to extend this notion to more that 2 actions. In this case, we would consider the $\argmax_{a_i \in A_i} \vec\theta_i \cdot (\overline{u}_{ij}(a_i,\sigma_{-i}))_{j \in J}$ to determine the action played by agent $i$ with type $\vec\theta_i$. When the highest value is reached by two or more actions $a_i$, the response of agent $i$ would be a mixed combination of such actions. This definition would still use the fact that a Nash Equilibrium solution (be it pure or mixed) must be made of threshold strategies.

\newpage

\end{document}